\documentclass[12pt,letter]{article}
\usepackage[utf8]{inputenc}
\usepackage{geometry}
\usepackage{setspace,amsmath,amsfonts,amssymb,amsthm,cite,tabularx,longtable,bbm,bm,eucal,caption,wrapfig,fancyhdr,caption,booktabs,graphicx,dsfont,csquotes,zref,lscape,algorithm, algorithmic,subcaption,newfloat,tikz,xcolor}
\usepackage{rotating}
\usepackage{tabularray}
\usepackage[backref = page]{hyperref}
\usepackage[square,sort,comma,longnamesfirst]{natbib}
\bibpunct{(}{)}{;}{a}{,}{,}

\usetikzlibrary{snakes}
\definecolor{newred}{HTML}{E66100}
\definecolor{newgreen}{HTML}{44AA99}
\definecolor{newyellow}{HTML}{E899D5}
\makeatother
\hypersetup{colorlinks,linkcolor={newred},citecolor={newgreen},urlcolor={newyellow}} 
\DeclareMathOperator*{\E}{\mathbb{E}}
\newcommand{\argmin}{\arg\!\min}

\date{September 23, 2023}
\usepackage[shortlabels]{enumitem}
\def\defeq{\equiv}
\newcommand{\1}[1]{\mathds{1}\left[#1\right]}
\newcommand{\EE}[2]{\mathbb{E}_{#1\!\!}\left[#2\right]}

\def\E#1{\EE{\,}{#1}}

\def\bE{{\mathbf E}}
\def\bI{{\mathbf I}}
\def\bZ{{\mathbf Z}}

\def\bA{{\mathbf A}}

\def\bf{{\mathbf f}}
\def\bg{{\mathbf g}}
\def\bF{{\mathbf F}}
\def\bu{{\mathbf u}}
\def\bv{{\mathbf v}}
\def\bw{{\mathbf w}}
\def\bB{{\mathbf B}}
\def\bQ{{\mathbf Q}}

\def\bV{{\mathbf V}}

\def\bU{{\mathbf U}}

\def\bS{{\mathbf S}}
\def\bW{{\mathbf W}}

\def\bC{{\mathbf C}}
\def\bx{{\mathbf x}}

\def\be{{\mathbf e}}
\def\bb{{\mathbf b}}
\def\by{{\mathbf y}}

\def\bs{{\mathbf s}}

\def\bLambda{{\bm \Lambda}}

\def\biota{{\bm \iota}}
\def\balpha{{\bm \alpha}}

\def\bxi{{\bm \xi}}

\def\bbeta{{\bm \beta}}

\def\btheta{{\bm \theta}}
\def\bzeta{{\bm \zeta}}

\def\bSigma{{\bm \Sigma}}
\def\bGamma{{\bm \Gamma}}
\def\bTheta{{\bm \Theta}}
\def\boldm{{\mathbf m}}

\def\bvarepsilon{{\bm \varepsilon}}
\usepackage{mathtools}
\usepackage{wasysym}
\usepackage{indentfirst}
\DeclareMathAlphabet{\mathcal}{OMS}{cmsy}{m}{n}
\newtheorem{remark}{Remark}

\DeclarePairedDelimiter\abs{\lvert}{\rvert}%
\DeclarePairedDelimiter\norm{\lVert}{\rVert}%

\makeatletter
\let\oldabs\abs
\def\abs{\@ifstar{\oldabs}{\oldabs*}}
\let\oldnorm\norm
\def\norm{\@ifstar{\oldnorm}{\oldnorm*}}
\makeatother

\theoremstyle{plain}
\newtheorem{thm}{Theorem}

\newcommand{\vertiii}[1]{{\left\vert\kern-0.25ex\left\vert\kern-0.25ex\left\vert #1 
		\right\vert\kern-0.25ex\right\vert\kern-0.25ex\right\vert}}
{
	\theoremstyle{plain}
	
}

\makeatletter
\def\mathcolor#1#{\@mathcolor{#1}}
\def\@mathcolor#1#2#3{%
	\protect\leavevmode
	\begingroup
	\color#1{#2}#3%
	\endgroup
}
\makeatother

\numberwithin{equation}{section}

\usepackage[toc, page, title]{appendix}
\newtheorem{lem}{Lemma}
\begin{document}	
	\title{\textbf{Combining Forecasts under Structural Breaks\\ Using Graphical LASSO}}
	\author{
		Tae-Hwy Lee\footnote{Department of Economics, University of California Riverside. Email: tae.lee@ucr.edu.}\hskip 4mm \ and \hskip 2mm
		Ekaterina Seregina\footnote{Department of Economics, Colby College. Email: eseregin@colby.edu. }\hskip 8mm 
	}
	\maketitle
	\thispagestyle{empty}
	
	\begin{abstract}
		\begin{spacing}{1}
In this paper we develop a novel method of combining many forecasts based on a machine learning algorithm called Graphical LASSO (GL). We visualize forecast errors from different forecasters as a network of interacting entities and generalize network inference in the presence of common factor structure and structural breaks. 
First, we note that forecasters often use common information and hence make common mistakes, which makes the forecast errors exhibit common factor structures. We use the Factor Graphical LASSO (FGL, \cite{fgl}) to separate common forecast errors from the idiosyncratic errors and exploit sparsity of the precision matrix of the latter.
Second, since the network of experts changes over time as a response to unstable environments such as recessions, it is unreasonable to assume constant forecast combination weights. Hence, we propose Regime-Dependent Factor Graphical LASSO (RD-FGL) that allows factor loadings and idiosyncratic precision matrix to be regime-dependent. We develop its scalable implementation using the Alternating Direction Method of Multipliers (ADMM) to estimate regime-dependent forecast combination weights.
The empirical application to forecasting macroeconomic series using the data of the European Central Bank's Survey of Professional Forecasters (ECB SPF) demonstrates superior performance of a combined forecast using FGL and RD-FGL.

	\end{spacing}

		\vskip 2mm
		\noindent \textit{Keywords}: Common Forecast Errors, Regime Dependent Forecast Combination, Sparse Precision Matrix of Idiosyncratic Errors, Structural Breaks.
		\vskip 2mm
		
		\noindent \textit{JEL Classifications}: C13, C38, C55
		
		\newpage
	\end{abstract} 
	
	\newpage 
	\setlength{\baselineskip}{22pt}
	\setcounter{page}{1}
		\setstretch{1.9}
	\section{Introduction}
	A search for the best forecast combination has been an important  on-going research question in economics. \cite{CLEMEN1989} pointed out that combining forecasts is \enquote{practical, economical and useful. Many empirical tests have demonstrated the value of composite forecasting. We no longer need to justify that methodology}. However, as demonstrated by \cite{DIEBOLD2018}, there are still some unresolved issues. Despite the findings based on the theoretical grounds, equal-weighted forecasts have proved surprisingly difficult to beat. Many methodologies that seek for the best forecast combination use equal weights as a benchmark: for instance, \cite{DIEBOLD2018} develop \enquote{partially egalitarian LASSO}.
	
	The success of equal weights is partly due to the fact that the forecasters use the same set of public information to make forecasts, hence, they tend to make common mistakes. For example, in the ECB SPF of Euro-area real GDP growth, the forecasters tend to \textit{jointly} understate or overstate GDP growth. Therefore, we stipulate that the forecast errors include common and idiosyncratic components, which allows the forecast errors to move together due to the common error component. Our paper provides a simple framework to learn from analyzing forecast errors: we separate unique errors from the common errors to improve the accuracy of the combined forecast.

	Dating back to \cite{GrangerBatesWeights}, the well-known expression for the optimal forecast combination weights requires an estimator of inverse covariance (precision) matrix. Precision matrix represents a network of interacting entities, such as corporations or genes. When the data is Gaussian, the sparsity in the precision matrix encodes the conditional independence graph - two variables are conditionally independent given the rest if and only if the entry corresponding to these variables in the precision matrix is equal to zero. Graphical models are a powerful tool to directly estimate precision matrix, avoiding the step of obtaining an estimator of covariance matrix to be inverted. Prominent examples of graphical models include GL (\cite{GLASSO}) and nodewise regression (\cite{meinshausen2006}). Despite using different strategies for estimating precision matrix, all graphical models assume that precision matrix is sparse: many entries of precision matrix are zero, which is a necessary condition to consistently estimate inverse covariance. Our paper demonstrates that such assumption contradicts the stylized fact that experts tend to make common mistakes and hence the forecast errors move together through common factors. \cite{fgl} show that graphical models fail to recover the entries of a nonsparse precision matrix under the factor structure and propose FGL that combines the benefits of graphical models and factor models.
	
	At the same time, the network of experts changes over time, that is, the relationships between forecasts produced by different experts or models can change either smoothly or abruptly (e.g., as a response to an unexpected policy shock, or in the times of economic downturns). Such changes give rise to different regimes and it is important to account for changes in optimal forecast combination weights induced by structural breaks. This paper augments \cite{fgl} and develops a unified framework to generalize network inference in the presence of structural breaks. As a first extension, we model structural changes in factor loadings. As a second extension, we model structural changes in the precision matrix of the idiosyncratic component after removing common factors. We estimate regime-dependent precision matrix for forecast combination using both pre- and post-break data when forecast errors are driven by common factors. We call the proposed algorithm \textit{Regime-Dependent Factor Graphical LASSO} (RD-FGL) and develop its scalable implementation using the Alternating Direction Method of Multipliers (ADMM).
	
	Our paper makes several contributions. First, we allow the forecast errors to be highly correlated  due to the common component which is motivated by the stylized fact that the forecasters tend to jointly understate or overstate the predicted series of interest. Second, to tackle changing relationships between forecasts produced by different experts or models as a response to unstable environments, we develop a unified framework to generalize network inference in the presence of structural breaks. We propose RD-FGL that models structural changes in factor loadings and idiosyncratic precision matrix. We develop scalable implementation of RD-FGL using ADMM to estimate regime-dependent forecast combination weights. Third, an empirical application to forecasting macroeconomic series using the data of the ECB SPF shows that incorporating (i) factor structure in the forecast errors together with (ii) sparsity in the precision matrix of the idiosyncratic components and (iii) regime-dependent combination weights improves the performance of a combined forecast over forecast combinations using equal weights.
	
	We emphasize that in this paper our goal is to develop a framework for forecast combinations that incorporates structural breaks which have already occurred in the past. We neither consider the possibility of breaks over the forecast horizon (as explored in \cite{pesaran2006hierarchical}), nor study the case of the out-of-sample breaks. These scenarios are interesting extensions of this paper, however they lie outside the scope of this paper. We also emphasize that we take the individual forecasts to be combined as given and do not discuss how the forecasts are generated.
	
	The paper is structured as follows. Section 2 studies the approximate factor model for the forecast errors. Section 3 reviews FGL and contains theoretical results on the consistency of the FGL estimator for forecast combinations. Section 4 introduces Regime-Dependent graphical model and discusses its implementation using ADMM. Section 5 validates theoretical results using simulations. Section 6 studies an empirical application for macroeconomic time-series forecasting. Section 7 concludes.
	
	\textbf{Notation}. For the convenience of the reader, we summarize the notation to be used throughout the paper. Let $\mathcal{S}_p$ denote the set of all $p \times p$ symmetric matrices. For any matrix $\bC$, its $(i,j)$-th element is denoted as $c_{ij}$. Given a vector $\bu\in \mathbb{R}^d$ and a parameter $a\in \lbrack1,\infty)$, let $\norm{\bu}_a$ denote $\ell_a$-norm. Given a matrix $\bU \in\mathcal{S}_p$, let $\lambda_{\text{max}}(\bU) \defeq \lambda_1(\bU) \geq \lambda_2(\bU)\geq \ldots \geq \lambda_{\text{min}}(\bU) \defeq \lambda_p(\bU)$ be the eigenvalues of $\bU$.
	 Given a matrix $\bU \in \mathbb{R}^{p\times p}$ and parameters  $a,b\in \lbrack1,\infty)$, let $\vertiii{\bU}_{a,b}\defeq \max_{\norm{\by}_a=1}\norm{\bU\by}_{b}$ denote the induced matrix-operator norm. The special cases are $\vertiii{\bU}_1\defeq \max_{1\leq j\leq p}\sum_{i=1}^{p}\abs{u_{ij}}$ for the $\ell_1/\ell_1$-operator norm;  the operator norm ($\ell_2$-matrix norm) $\vertiii{\bU}_{2}^{2}\defeq\lambda_{\text{max}}(\bU\bU')$ is equal to the maximal singular value of $\bU$.
 Finally, $\norm{\bU}_{\text{max}}\defeq\max_{i,j}\abs{u_{ij}}$ denotes the element-wise maximum. For two sequences $a_{T,p}$ and $b_{T,p}$, we denote $a_{T,p}\asymp b_{T,p}$ if there exist constants $c_1, c_2 >0$ such that $c_1 a_{T,p} \leq b_{T,p} \leq c_2 a_{T,p}$.

	\section{Approximate Factor Models for Forecast Errors}	
	The approximate factor models for the forecasts were first considered by \cite{StockChan1999}. They modeled a panel of ex-ante forecasts of a single time-series as a dynamic factor model and found out that the combined forecasts improved on individual ones when all forecasts have the same information set (up to difference in lags). This result emphasizes the benefit of forecast combination even when the individual forecasts are not based on different information and, therefore, do not broaden the information set used by any one forecaster.
	
	In this paper, we are interested in finding the combination of forecasts which yields the best out-of-sample performance in terms of the mean-squared forecast error. We claim that the forecasters use the same set of public information to make forecasts and hence they tend to make common mistakes. Figure \ref{fig1} illustrates this statement: it shows quarterly forecasts of Euro-area real GDP growth produced by the ECB SPF from 1999Q3 to 2019Q3. As described in \cite{DIEBOLD2018}, forecasts are solicited for one year ahead of the latest available outcome: e.g., the 2007Q1 survey asked the respondents to forecast the GDP growth over 2006Q3-2007Q3. As evidenced from Figure \ref{fig1}, forecasters tend to jointly understate or overstate GDP growth, meaning that their forecast errors include common and idiosyncratic parts. Therefore, we can model the tendency of the forecast errors to move together via factor decomposition.
	\begin{figure}[!htb]
	\centering
	\includegraphics[width=.95\linewidth]{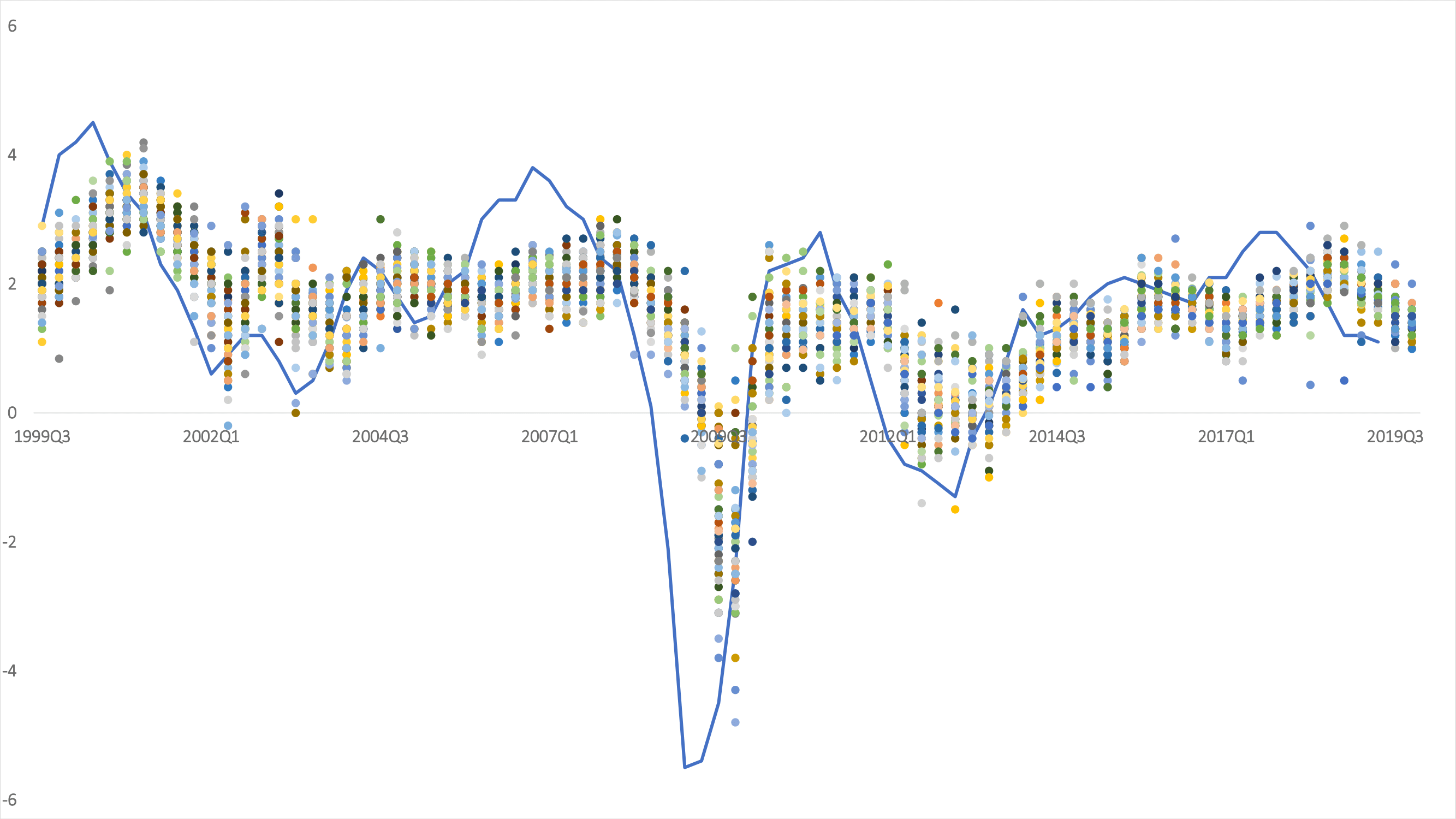}
	\bigskip
	\caption{\textbf{The European Central Bank's (ECB) Survey of Professional Forecasters (SPF)}. Each circle denotes the forecast of each professional forecaster in the SPF for the quarterly 1-year-ahead forecasts of Euro-area real GDP growth, year-on-year percentage change. Actual series is the blue line. \textit{Source: \href{https://www.ecb.europa.eu/stats/ecb_surveys/survey_of_professional_forecasters/html/index.en.html}{European Central Bank}.}}
	\bigskip
	\label{fig1}
\end{figure}

	Suppose we have $p$ competing forecasts of the univariate series $y_t$,  $t=1,\ldots,T$ and $\widetilde{\be}_t=(\widetilde{e}_{1t},\ldots,\widetilde{e}_{pt})' \sim \mathcal{N} (\boldm, \bSigma)$ is a $p \times 1$ vector of forecast errors. Note that we allow bias in the forecasts. In fact, Figure \ref{fig1} demonstrates that the individual forecasts are indeed biased. Assume that the generating process for the forecast errors follows a $q$-factor model: $\widetilde{\be}_t = \boldm + \bB\bf_{t}+\bvarepsilon_{t}$,
	where $\bf_t=(f_{1t},\ldots, f_{qt})'$ are the common factors of the forecast errors for $p$ models, $\bB$ is a $p \times q$ matrix of factor loadings, and $\bvarepsilon_t$ is the idiosyncratic component that cannot be explained by the common factors. Define demeaned forecast errors as $
	\be_t \defeq \widetilde{\be}_t - \boldm$ such that: 
	\begin{align} \label{equ1}
	&\underbrace{\be_t}_{p \times 1}=\bB \underbrace{\bf_t}_{q\times 1}+ \ \bvarepsilon_t,\quad t=1,\ldots,T,
	\end{align}

	 Unobservable factors, $\bf_{t}$, and loadings, $\bB$, are usually estimated by the principal component analysis (PCA), studied in \cite{Bai2002,Stock2002}. Strict factor structure assumes that the idiosyncratic forecast error terms, $\bvarepsilon_{t}$, are uncorrelated with each other, whereas approximate factor structure allows correlation of the idiosyncratic components (\cite{Chamberlain}).

	We use the following notations: $\E{\bvarepsilon_t\bvarepsilon'_t}=\bSigma_{\varepsilon}$, $\E{\bf_t\bf'_t}=\bSigma_{f}$, and $\E{\be_t\be'_t}=\bSigma=\bB\bSigma_{f}\bB'+ \bSigma_{\varepsilon}$. Let $\bTheta=\bSigma^{-1}$, $\bTheta_{\varepsilon}=\bSigma_{\varepsilon}^{-1}$ and $\bTheta_{f}=\bSigma_{f}^{-1}$ be the precision matrices of forecast errors, idiosyncratic and common components respectively.
	The objective function to recover factors and loadings from \eqref{equ1} is:
	\begin{align} \label{eq7}
	\min_{\bf_1,\ldots, \bf_T, \bB}\frac{1}{T}\sum_{t=1}^{T}(\be_t-\bB\bf_t)'(\be_t-\bB\bf_t),\ \text{s.t.} \ \bB'\bB=\bI_q,
	\end{align}
	where the constraint is necessary for the unique identification of factors. Fixing the value of $\bB$, we can project forecast errors $\be_t$ into the space spanned by $\bB$: $\bf_t=(\bB'\bB)^{-1}\bB'\be_t=\bB'\be_t$. When combined with \eqref{eq7}, this yields a concentrated objective function for $\bB$:
	\begin{equation} \label{equ10}
	\max_{\bB} \ \text{tr}\Big[\bB'\Big(\frac{1}{T}\sum_{t=1}^{T}\be_t\be_t'\Big) \bB\Big].
	\end{equation}
	It is well-known (see \cite{Stock2002} among others) that $\widehat{\bB}$ estimated from the first $q$ eigenvectors of $\frac{1}{T}\sum_{t=1}^{T}\be_t\be_t'$ is the solution to \eqref{equ10}.
	Given a sample of the estimated residuals $\{\widehat{\bvarepsilon}_t=\be_t-\widehat{\bB}\widehat{\bf_t}\}_{t=1}^{T}$ and the estimated factors $\{\widehat{\bf}_t\}_{t=1}^{T}$, let $\widehat{\bSigma}_{\varepsilon} = (1/T)\sum_{t=1}^{T}\widehat{\bvarepsilon}_t\widehat{\bvarepsilon}_t'$ and $\widehat{\bSigma}_{f}=(1/T)\sum_{t=1}^{T}\widehat{\bf}_t\widehat{\bf}_t'$ be the sample counterparts of the covariance matrices.
	
	Moving forward to the forecast combination exercise, suppose we have $p$ competing forecasts, $\widehat{\by}_{t}=(\hat{y}_{1,t},\ldots,\hat{y}_{p,t})'$, of the variable $y_t$, $t=1,\ldots,T$. The forecast combination is defined as $\widehat{y}_{t}^{c}=\bw'\widehat{\by}_{t}$, where $\bw$ is a $p \times 1$ vector of weights. Define a measure of risk $\text{MSFE}(\bw, \bSigma)=\bw'\bSigma\bw$. As shown in \cite{GrangerBatesWeights}, the \textit{optimal} forecast combination minimizes the MSFE of the combined forecast error:
	\begin{equation} \label{equ11}
	 \min_{\bw} \text{MSFE}=\min_{\bw}\E{\bw'\be_t\be^{'}_{t}\bw}=\min_{\bw} \bw'\bSigma\bw, \ \text{s.t.} \ \bw'\biota_p =1,
	\end{equation}
	where $\biota_p$ is a $p\times 1$ vector of ones.\footnote{As noted in \cite{timmerman_handbook_of_forecasting}, global mean-variance optimization problem in finance is similar to the forecast combination problem. Combined forecasts can be viewed as the portfolio and the source of risk reflects incomplete information about the target variable and model
		misspecification possibly due to non-stationarities in the underlying data generating process.} The solution to \eqref{equ11} yields a $p\times 1$ vector of the optimal forecast combination weights:
	\begin{equation} \label{eq13}
	\bw=\frac{\bTheta\biota_p}{\biota_p'\bTheta\biota_p}.
	\end{equation}
 	If the true precision matrix is known, the equation \eqref{eq13} guarantees to yield the optimal forecast combination. In reality, one has to estimate $\bTheta$. As pointed out by \cite{smith2009simple}, when the estimation uncertainty of the weights is taken into account, there is no guarantee that the \enquote{optimal} forecast combination will be better than the equal weights or even improve the individual forecasts. Define $a = \biota'_{p}\bTheta\biota_p/p$ and $\widehat{a} = \biota'_{p}\widehat{\bTheta}\biota_p/p$. We can write $\abs{\frac{\text{MSFE}(\widehat{\bw},\widehat{\bSigma})}{\text{MSFE}(\bw,\bSigma)} -1} = \abs{ \frac{\hat{a}^{-1} }{a^{-1}}-1}=\frac{\abs{a-\hat{a}}}{\abs{\hat{a}}}$ and $\norm{\widehat{\bw}-\bw}_1 = \Big[(a\widehat{\bTheta}\biota_p) - (a\bTheta\biota_p) + (a\bTheta\biota_p) - (\hat{a}\bTheta\biota_p)   \Big]/p \cdot (\hat{a}a)$.
	Therefore, in order to control the estimation uncertainty in the MSFE and combination weights, one needs to obtain a consistent estimator of the precision matrix $\bTheta$. More details are discussed in Section 3 and Theorem \ref{theor1}.
	
	\section{Factor Graphical LASSO for Forecast Errors}	
	Since our interest is in constructing weights for the forecast combination, our goal is to estimate a precision matrix of the forecast errors. This brings us to consider a family of graphical models, which have evolved from the connection between partial correlations and the entries of an adjacency matrix. The adjacency matrix has zero or one in its entries, with a zero entry indicating that two variables are independent conditional on the rest. The adjacency matrix is sometimes referred to as a ``graph". In graphical models, each vertex represents a random variable, and the graph visualizes the joint distribution of the entire set of random variables.
	\textit{Sparse graphs} have a relatively small number of edges. Weighted Graphical Lasso (GL) procedure (\cite{GLASSO}) described in Supplemental Appendix \ref{appendixA0theor} is a representative member of graphical models family.
	
	Before estimating precision matrix of forecast errors, $\bTheta$, we first obtain estimates of factors, $\widehat{\bf}_t$, and factor loadings, $\widehat{\bB}$, using PCA. Second, we obtain $\widehat{\bSigma}_{f}=\frac{1}{T}\sum_{t=1}^{T}\widehat{\bf}_{t}\widehat{\bf}_{t}^{'}$, $\widehat{\bTheta}_{f}=\widehat{\bSigma}_{f}^{-1}$, $\widehat{\bvarepsilon}_t = \be_t-\widehat{\bB}\widehat{\bf_t}$, and $\widehat{\bSigma}_{\varepsilon}=\frac{1}{T}\sum_{t=1}^{T}\widehat{\bvarepsilon}_{t}\widehat{\bvarepsilon}_{t}^{'}$. Third, we note that when common factors are present across the forecast errors, the precision of forecast errors, $\bTheta$, cannot be sparse because all pairs of the forecast errors are partially correlated given other forecast errors through the common factors. Therefore, instead of imposing sparsity assumption on $\bTheta$ we require sparsity of the precision matrix of the idiosyncratic errors, $\bTheta_{\varepsilon}$. 
	
	 Once we condition on the common components, it is sensible to assume that many remaining partial correlations of $\bvarepsilon_{t}$ will be negligible and thus $\bTheta_{\varepsilon}$ is sparse. Hence, it is estimated with the Weighted Graphical Lasso penalty:
	\begin{align} \label{e7.6_main}
		&\widehat{\bTheta}_{\varepsilon,\tau}=\argmin_{\bTheta_{\varepsilon}=\bTheta'_{\varepsilon}}\text{tr}(\bW_{\varepsilon}\bTheta_{\varepsilon})-\log\det(\bTheta_{\varepsilon})+\tau\sum_{i\neq j}\widehat{\gamma}_{\varepsilon,ii}\widehat{\gamma}_{\varepsilon,jj}\abs{\theta_{\varepsilon,ij}},
	\end{align}
	initialized with $\bW_{\varepsilon}=\widehat{\bSigma}_{\varepsilon}+\tau\bI$, where  $\widehat{\gamma}_{\varepsilon,ii}$ is the $(i,i)$-th element of $\widehat{\bGamma}_{\varepsilon}^2\defeq \textup{diag}(\bW_{\varepsilon})$. The subscript $\tau$ in $\widehat{\bTheta}_{\varepsilon,\tau}$ means that the solution of the optimization problem in \eqref{e7.6_main} will depend upon the choice of the tuning parameter $\tau$. In order to simplify notation, we will omit the subscript $\tau$.
	
	Finally we put all estimates together using the Sherman-Morrison-Woodbury formula to estimate the precision of forecast errors:
	\begin{equation}\label{3.11_main}
		\widehat{\bTheta}=\widehat{\bTheta}_{\varepsilon}-\widehat{\bTheta}_{\varepsilon}\widehat{\bB}\lbrack\widehat{\bTheta}_f+\widehat{\bB}'\widehat{\bTheta}_{\varepsilon}\widehat{\bB}\rbrack^{-1}\widehat{\bB}'\widehat{\bTheta}_{\varepsilon}.
	\end{equation}
	
	The aforementioned procedure introduced in \cite{fgl} is called Factor Graphical LASSO (FGL).\footnote{Instead of decomposing covariance matrix into low-rank and idiosyncratic components, $\E{\be_t\be'_t}=\bSigma=\bB\bSigma_{f}\bB'+ \bSigma_{\varepsilon}$, \cite{shi2020l2} regularize $\bSigma$ directly with the factor structure (more broadly, latent group structure) in mind. FGL approach is different since it focuses on estimating precision matrix directly which is the main focus of the theoretical analysis in \cite{fgl}.
	} It is summarized in Supplemental Appendix \ref{appendix_fgl}, where we also discuss the choice of the tuning parameter $\tau$ in \eqref{e7.6_main}.

	We can use $\widehat{\bTheta}$ to estimate the forecast combination weights $	\widehat{\bw}=\bTheta\biota_p/\biota_p'\widehat{\bTheta}\biota_p$. This approach allows to extract the benefits of modeling common movements in forecast errors, captured by a factor model, and the benefits of using many competing forecasting models that give rise to a high-dimensional precision matrix, captured by a graphical model.

	Let us now examine the asymptotic properties of FGL. We first introduce some terminology and notations. Let $A\in \mathcal{S}_p$. Define the following set for $j=1,\ldots,p$:
	\begin{align}\label{equ84}
		&D_j(A)\defeq\{ i:A_{ij}\neq 0,\ i\neq j\}, \quad d_j(A)\defeq\text{card}(D_j(A)),\quad d(A)\defeq\max_{j=1,\ldots,p}d_j(A),
	\end{align}
	where $d_j(A)$ is the number of edges adjacent to the vertex $j$ (i.e., the \textit{degree} of vertex $j$), and $d(A)$ measures the maximum vertex degree. Define $S(A)\defeq \bigcup_{j=1}^{p}D_j(A)$ to be the overall off-diagonal sparsity pattern, and $s(A)\defeq \sum_{j=1}^{p}d_j(A)$ is the overall number of edges contained in the graph. Note that $\text{card}(S(A)) \leq s(A)$: when $s(A)=p(p-1)/2$ this would give a fully connected graph.
	We now list the assumptions on the model \eqref{equ1}:
	\begin{enumerate}[\textbf{({A}.1)}]
		\item \label{A1} (Spiked covariance model)
		Assume that (i) As $p \rightarrow \infty$, $\lambda_1(\bSigma)>\lambda_2(\bSigma)>\ldots>\lambda_q(\bSigma)\gg \lambda_{q+1}(\bSigma)\geq \ldots \geq \lambda_p(\bSigma) > 0$, where $\lambda_j(\bSigma)=\mathcal{O}(p)$ for $j \leq q$, while the non-spiked eigenvalues are bounded, that is, $c_0 \leq \lambda_j(\bSigma) \leq C_0$, $j > q$ for constants $c_0, C_0 > 0$.
		And assume that (ii) $\biota'_p\bTheta\biota_p/p \geq c >0$, where $c>0$ is a positive constant.
	\end{enumerate}
	\begin{enumerate}[\textbf{({A}.2)}]
		\item \label{A2}(Pervasive factors)
		There exists a positive definite $q \times q$ matrix $\breve{\bB}$ such that\\ $\vertiii{p^{-1}\bB'\bB-\breve{\bB}}_{2}\rightarrow 0$ and $\lambda_{\text{min}}(\breve{\bB})^{-1}=\mathcal{O}(1)$  as $p \rightarrow \infty$.
	\end{enumerate}
	We also impose strong mixing condition. Let $\mathcal{F}_{-\infty}^{0}$ and $\mathcal{F}_{T}^{\infty}$ denote the $\sigma$-algebras that are generated by $\{(\bf_t,\bvarepsilon_{t}):t\leq 0\}$ and $\{(\bf_t,\bvarepsilon_{t}):t\geq T\}$ respectively. Define the mixing coefficient $		\alpha(T)=\sup_{A\in \mathcal{F}_{-\infty}^{0}, B \in \mathcal{F}_{T}^{\infty}}\abs{\Pr{A}\Pr{B}-\Pr{AB}}$.
	\begin{enumerate}[\textbf{({A}.3)}]
		\item \label{A3} (Strong mixing) There exists $r_3>0$ such that $3r_{1}^{-1}+1.5r_{2}^{-1}+3r_{3}^{-1}>1$, and $C>0$ satisfying, for all $T\in \mathbb{Z}^{+}$, $\alpha(T)\leq \exp (-CT^{r_3})$.
	\end{enumerate}	
	Assumption \ref{A1} divides the eigenvalues into the diverging and bounded ones. This assumption is satisfied by the factor model with pervasive factors, which is stated in Assumption \ref{A2}. We say that a factor is pervasive in the sense that it has non-negligible effect on a non-vanishing proportion of individual time-series. Part (ii) of Assumption \ref{A1} is needed for consistent estimation of the optimal forecast combination weights.  Assumptions \ref{A1}-\ref{A2} are crucial for estimating a high-dimensional factor model: they ensure that the space spanned by the principal components in the population level $\bSigma$ is close to the space spanned by the columns of the factor loading matrix $\bB$. Assumption \ref{A3} is a technical condition which is needed to consistently estimate the factors and loadings. Assumptions \ref{A1}(i), \ref{A2}, and \ref{A3} are standard assumptions and are used in \cite{fan2013POET}.
	
	Let $\bLambda_q = \text{diag}(\lambda_{1},\ldots,\lambda_q)$ be a matrix of $q$ leading eigenvalues of $\bSigma$, and $\bV_q = (\bv_1,\ldots,\bv_q)$ is a $p\times q$ matrix of their corresponding leading eigenvectors. Define $\widehat{\bSigma}, \widehat{\bLambda}_q,\widehat{\bV}_q$ to be the estimators of $\bSigma,\bLambda_q,\bV_q$. We further let $\widehat{\bLambda}_q=\text{diag}(\hat{\lambda}_1,\ldots,\hat{\lambda}_q)$ and $\widehat{\bV}_q=(\hat{\bv}_1,\ldots,\hat{\bv}_q)$ to be constructed by the first $q$ leading empirical eigenvalues and the corresponding eigenvectors of $\widehat{\bSigma}$ and $\widehat{\bB}\widehat{\bB}'=\widehat{\bV}_q\widehat{\bLambda}_q\widehat{\bV}_{q}^{'}$. Similarly to \cite{fan2018elliptical}, we require the following bounds on the componentwise maximums of the estimators:
	\begin{enumerate}[\textbf{({B}.1)}]
		\item \label{B1} $\norm{\widehat{\bSigma}-\bSigma}_{\text{max}}=\mathcal{O}_P(\sqrt{\log p/T})$,
	\end{enumerate}
	
	\begin{enumerate}[\textbf{({B}.2)}]
		\item \label{B2} $\norm{(\widehat{\bLambda}_q-\bLambda_q)\bLambda_{q}^{-1}}_{\text{max}}=\mathcal{O}_P(\sqrt{\log p/T})$,
	\end{enumerate}
	\begin{enumerate}[\textbf{({B}.3)}]
		\item \label{B3} $\norm{\widehat{\bV}_q-\bV_q}_{\text{max}}=\mathcal{O}_P(\sqrt{\log p/(Tp)})$.
	\end{enumerate}
	Assumptions \ref{B1}-\ref{B3} are needed in order to ensure that the first $q$ principal components are approximately the same as the columns of the factor loadings. The estimator $\widehat{\bSigma}$ can be thought of as any ``pilot" estimator that satisfies \ref{B1}. For sub-Gaussian distributions, sample covariance matrix, its eigenvectors and eigenvalues satisfy \ref{B1}-\ref{B3}.
	
	In addition, the following structural assumption on the model is imposed:
	\begin{enumerate}[\textbf{({C}.1)}]
		\item \label{C1} $\norm{\bSigma}_{\text{max}}=\mathcal{O}(1)$ and $\norm{\bB}_{\text{max}}=\mathcal{O}(1)$.
	\end{enumerate}
	Note that Assumptions \ref{B1}-\ref{B3} and \ref{C1} are standard assumptions and are used in \cite{fan2018elliptical}.
	
	To study the properties of the combination weights and MSFE, we first need to establish the convergence properties of precision matrix produced by Algorithm \ref{alg2}. Let $\omega_{T}\defeq \sqrt{\log p/T} +1/\sqrt{p}$. Also, let $s(\bTheta_{\varepsilon})=\mathcal{O}_P(s_T)$ for some sequence $s_T\in(0,\infty)$ and $d(\bTheta_{\varepsilon})=\mathcal{O}_P(d_T)$ for some sequence $d_T\in(0,\infty)$. The deterministic sequences $s_T$ and $d_T$ will control the sparsity $\bTheta_{\varepsilon}$ for FGL. Note that $d_T$ can be smaller than or equal to $s_T$. 
	
	Let $\varrho_{T}$ be a sequence of positive-valued random variables such that $\varrho_{T}^{-1}\omega_{T}\xrightarrow{P}0$ and $\varrho_{T}d_Ts_T\xrightarrow{P}0$, with $\tau \asymp \omega_{T}$ (where $\tau$ is the tuning parameter for the FGL in \eqref{e7.6}). \cite{fgl} show that under the Assumptions \ref{A1}-\ref{A3}, \ref{B1}-\ref{B3} and \ref{C1}, $\vertiii{\widehat{\bTheta}-\bTheta}_{1}=\mathcal{O}_P(\varrho_{T}d_Ts_T) = o_P(1)$ and $\vertiii{\widehat{\bTheta}-\bTheta}_{2}=\mathcal{O}_P(\varrho_{T}s_T) = o_P(1)$ for FGL.
	
	Having established the convergence rates for precision matrix, we now study the properties of the combination weights and the resulted MSFE.
	\begin{thm} \label{theor1}
		Assume \ref{A1}-\ref{A3}, \ref{B1}-\ref{B3}, and \ref{C1} hold. FGL consistently estimates forecast combination weights and $\text{MSFE}(\widehat{\bw},\widehat{\bSigma})$:
		\begin{enumerate}
			\item[(i)]If $\varrho_{T}d_{T}^{2}s_T\xrightarrow{\text{P}}0$, $\norm{\widehat{\bw}-\bw}_1=\mathcal{O}_P\Big(\varrho_{T}d_{T}^2s_T\Big)=o_P(1)$.
			\item[(ii)] If $\varrho_{T}d_Ts_T\xrightarrow{\text{P}}0$, $	\abs{\frac{\text{MSFE}(\widehat{\bw},\widehat{\bSigma})}{\text{MSFE}(\bw,\bSigma)} -1}=\mathcal{O}_P(\varrho_{T}d_Ts_T )=o_P(1)$.
		\end{enumerate}
	\end{thm}
	The proof of Theorem \ref{theor1} can be found in Supplementary Appendix \ref{appendixAtheor}. Note that the rates of convergence for MSFE and precision matrix $\bTheta$, which was derived in \cite{fgl}, are the same and both are faster than the combination weight rates. In contrast to the classical graphical model in Algorithm \ref{alg1a}, the convergence properties of which were examined by \cite{Sara2018} among others, the rates in Theorem \ref{theor1} depend on the sparsity of $\bTheta_{\varepsilon}$ rather than of $\bTheta$. This means that instead of assuming that many partial correlations of forecast errors $\be_{t}$ are negligible, which is not realistic under the factor structure, we impose a milder restriction requiring many partial correlations of $\bvarepsilon_{t}$ to be negligible once the common components have been taken into account.
	\section{RD-FGL for Forecast Errors}
	There are two streams of literature that study time-varying networks. The first one models dynamics in the precision matrix locally. \cite{Zhou2010} develop a nonparametric method for estimating time-varying graphical structure for multivariate Gaussian distributions using an $\ell_1$-penalized log-likelihood. They find out that if the covariances change smoothly over time, the covariance matrix can be estimated well in terms of predictive risk even in high-dimensional problems. \cite{Kolar2018} introduce nonparanormal graphical models that allow to model high-dimensional heavy-tailed systems and the evolution of their network structure. They show that the estimator consistently estimates the latent inverse Pearson correlation matrix. The second stream of literature allows the network to vary with time by introducing two different frequencies. \cite{tvgl} study time-varying Graphical LASSO with smoothing evolutionary penalty.
	
	We augment the framework in Section 3 to account for regime switching by modeling the change in precision matrix due to $N$ structural breaks. Define $n_j \defeq t_j - t_{j-1}$ to be the sample between the $j$-th and $(j-1)$-th break points, where $j=1,\ldots, N+1$, $\sum_{j=1}^{N+1}n_j = T$, $t_0 = 0$, $N\leq T$.
	\subsection{Regime-Dependent Factor Loadings}
	Macroeconomic and financial datasets typically span a long time period, hence, the assumption of time-invariant factor loadings is restrictive. As a first extension to FGL, we model structural changes in factor loadings using a framework similar to \cite{su2017time}. For now assume a single known break $N=1$ which occurs at $T_1$.\footnote{The possible presence of break over the forecast horizon, as explored in \cite{pesaran2006hierarchical}, would result in common forecast errors of all forecasting models, which can be captured by using the factor model. This would be another scenario of common forecast errors as motivated by Figure \ref{fig1}.} Write equation \eqref{equ1} as:	
	\begin{align} \label{equ1_tv}
		& e_{it}=\underbrace{\bb'_{i}}_{1\times q} \underbrace{\bf_t}_{q\times 1}+ \ \varepsilon_{it},\quad t=1,\ldots,T, \ i=1,\ldots,p.
	\end{align}
To estimate $\{\bb_{i}\}_{i=1}^{p}$ and $\{\bf_t\}_{t=1}^{T}$, we can consider the following weighted least squares problem:
\begin{equation}\label{wls}
	\min_{\{\bb_{i}\}_{i=1}^{p}, \{\bf_t\}_{t=1}^{T}} (pT)^{-1}\sum_{i=1}^{p}\sum_{t=1}^{T}\big[ e_{it} - \bb'_{i}\bf_t  \big]^2 K_{\gamma t},
\end{equation}
subject to certain identification restrictions to be specified later on. Here, $K_{\gamma t} = \gamma \1{t \leq T_1} + \1{t > T_1}$ is a discrete kernel as in \cite{li2013categorical} with $\gamma \in [0,1]$. Since more recent information is usually more relevant to forecasting, such kernel-weight estimator gives weight 1 to post-break observations and weight $\gamma$ to pre-break observations.\footnote{Adjusting the parameter to control the degree to which pre-break data are discounted by the model is useful for applied researchers. Bayesian framework with hierarchical priors provides an alternative way to combine the sample information contained in the objective function with prior information about the values of the model parameters and the relations among them (see \cite{pastor2001hierarchical,pesaran2006hierarchical}).}
 
Define the $T \times p$ matrices $\bE(\gamma) = \Big(\be_1(\gamma),\ldots, \be_p(\gamma)   \Big)$, $\mathcal{E}(\gamma) = \Big(\bvarepsilon_1(\gamma),\ldots, \bvarepsilon_p(\gamma)   \Big)$, where\\ $\be_{i}(\gamma) = \Big(K_{\gamma 1}^{1/2}e_{i1},\ldots, K_{\gamma T}^{1/2}e_{iT}\Big)'$ and $\bvarepsilon_{i}(\gamma) = \Big(K_{\gamma 1}^{1/2}\varepsilon_{i1},\ldots, K_{\gamma T}^{1/2}\varepsilon_{iT}\Big)'$. Also, let\\ $\bF(\gamma) = \Big(K_{\gamma 1}^{1/2}\bf_{1},\ldots, K_{\gamma T}^{1/2}\bf_{T}\Big)'$ be a $T \times q$ matrix collecting factors. In matrix notation, the transformed model in \eqref{equ1_tv} can be written as $\bE(\gamma) = \bF(\gamma)\bB' + \mathcal{E}(\gamma)$,
where $\bB= (\bb_{1},\ldots, \bb_{p})'$ is a $p \times q$ matrix of factor loadings.

As shown in \cite{su2017time} for the continuous kernel, the minimization problem in \eqref{wls} reduces to:
\begin{align}\label{wls-modified}
	&\min_{\bF(\gamma),\bB}  \text{tr} \Bigg[ \Big( \bE(\gamma) -  \bF(\gamma)\bB' \Big) \Big( \bE(\gamma) -  \bF(\gamma)\bB'  \Big)'  \Bigg]\\
	& \text{s.t.} \ \bF'(\gamma) \bF(\gamma) / T = \bI_q \ \text{and} \ \bB' \bB = \text{diagonal matrix}. \nonumber
\end{align}
The problem in \eqref{wls-modified} is the conventional PCA problem. The estimated factor matrix $\widehat{\bF}(\gamma) = \Big(K_{\gamma 1}^{1/2}\widehat{\bf}_{1},\ldots, K_{\gamma T}^{1/2}\widehat{\bf}_{T}\Big)'$ is $\sqrt{T}$ times eigenvectors corresponding to the $q$ largest eigenvalues of $\bE(\gamma)\bE'(\gamma)$, arranged in descending order, and $\widehat{\bB}'(\gamma) = (\widehat{\bF}(\gamma)\widehat{\bF}'(\gamma))^{-1}\widehat{\bF}'(\gamma)\bE(\gamma) = \widehat{\bF}'(\gamma) \bE(\gamma)/T $ are the estimators of the corresponding time-varying factor loadings, where $\widehat{\bB}(\gamma) = (\widehat{\bb}_{1}(\gamma),\ldots, \widehat{\bb}_{p}(\gamma))'$ is $p \times q$.

Since the estimator $\widehat{\bF}(\gamma)$ is only consistent up to a rotation, we use a two-stage estimation procedure to obtain a consistent estimator (\cite{su2017time}). Based on the consistent estimators of $\bb_{i}$'s obtained from the first stage, consistent estimators of $\bf_{t}(\gamma)$ can be obtained by considering the following least squares problem $\widehat{\bf}_t(\gamma) = \argmin_{\bf_t} \sum_{i=1}^{p} \Big[e_{it} - \widehat{\bb}^{'}_{i}(\gamma)\bf_{t}  \Big]^2$ which yields the solution $\widehat{\bf}_t(\gamma) = \Big( \sum_{i=1}^{p} \widehat{\bb}_{i}(\gamma) \widehat{\bb}^{'}_{i}(\gamma) \Big)^{-1}\Big(  \sum_{i=1}^{p} \widehat{\bb}_{i}(\gamma) e_{it} \Big)$.

As in \cite{su2017time}, we assume that $\E{\bf_t\bf'_t}$ is homogeneous over $t$. This assumption is not restrictive, since if $\E{\bf_t\bf'_t} = \bSigma_{f,t}$, we can rewrite the common component as $\bb'_{i}\bf_t = \Big( \bSigma_{f}^{-1/2}\bSigma_{f,t}^{1/2}\bb_{i} \Big)' \bSigma_{f}^{1/2} \bSigma_{f,t}^{-1/2}\bf_{t} = \bb^{*'}_{i}\bf^{*'}_{t}$, where $\bb^{*}_{i} = \bSigma_{f}^{-1/2}\bSigma_{f,t}^{1/2}\bb_{i} $, and $\bf^{*}_{t} = \bSigma_{f}^{1/2} \bSigma_{f,t}^{-1/2}\bf_{t}$ satisfies $\E{\bf^{*}_{t}\bf^{*'}_{t}} =\bSigma_{f} $ for each $t$. 

To choose the optimal tuning parameter $\gamma$ in \eqref{wls}, we use the cross-validation and solve the following minimization problem:
\begin{equation} \label{cv_gamma}
	\min_{\gamma} \text{CV}(\gamma) = \frac{1}{p(T-T_1)} \sum_{i=1}^{p}\sum_{s=T_1+1}^{T}\big[ e_{is} - \widehat{\bb}^{' (-s)}_{i}(\gamma) \widehat{\bf}^{(-s)}_{s}(\gamma)  \big]^2,
\end{equation}
where $\widehat{\bb}^{(-s)}_{i}(\gamma)$ and $\widehat{\bf}^{(-s)}_{s}(\gamma)$ are estimated by leaving the $s$-th time series observation out of the PCA procedure. 
\begin{remark}
	The procedure for estimating regime-dependent factor loadings can be easily extended to the case when the number of breaks is greater than 1 ($N>1$). The kernel in \eqref{wls} would be adjusted accordingly $K_{\gamma_j t} = \gamma_j \1{t \leq T_j} + \1{t > T_{N}}$, where $j=1,\ldots, N+1$. To estimate $\{\gamma_j\}_{j=1}^{N+1}$ we use cross-validation as in \eqref{cv_gamma} consequently applied to each two periods separated by a break.
\end{remark} 
	\subsection{Regime-Dependent Idiosyncratic Precision Matrix}
	
	 As a second  extension to FGL, we model structural changes in the precision matrix of the idiosyncratic component. Let $\bSigma_{\varepsilon,j}$ and $\bSigma_j$ be covariance matrices of idiosyncratic part and forecast errors in regime $j$. Define the corresponding precision matrices to be $\bTheta_{\varepsilon,j}\defeq\bSigma_{\varepsilon,j}^{-1}$ and $\bTheta_j \defeq \bSigma_{j}^{-1}$. Similarly to the previous subsection, without loss of generality we assume $\bSigma_{f_j}=\bSigma_f$ for all regimes $j$.
	
	 Let $\widehat{\bSigma}_{\varepsilon,j}=\frac{1}{n_j}\sum_{k=1}^{n_j}\widehat{\bvarepsilon}_{j,k}\widehat{\bvarepsilon}_{j,k}'$. To model dynamics in $\{\bTheta_{\varepsilon,j}\}_{j=1}^{N+1}$ we use the following optimization problem:
	\begin{align} \label{eq2}
		\min_{\{\bTheta_{\varepsilon,j}\}_{j=1}^{N+1}}\sum_{j=1}^{{N+1}}n_j\Big[\text{tr}\Big(\widehat{\bSigma}_{\varepsilon,j}\bTheta_{\varepsilon,j}\Big)-\log\det\bTheta_{\varepsilon,j}		\Big] 
		+\alpha\norm{\bTheta_{\varepsilon,j}}_{\text{od},1} 
		+\beta\sum_{j=2}^{{N+1}}\psi(\bTheta_{\varepsilon,j}-\bTheta_{\varepsilon,j-1}),
	\end{align}
	where the penalty for the off-diagonal (od) elements is $\norm{\bTheta_{\varepsilon,j}}_{\text{od},1}  = \sum_{l\neq q}\widehat{\gamma}_{\varepsilon,ll,j}\widehat{\gamma}_{\varepsilon,qq,j}\abs{\theta_{\varepsilon,lq,j}}$, $\widehat{\gamma}_{\varepsilon,ll,j}$ is the $(l,l)$-th element of $\widehat{\bGamma}_{\varepsilon,j}^2\defeq \textup{diag}(\widehat{\bSigma}_{\varepsilon,j})$ and $\theta_{\varepsilon,lq,j}$ is the $lq$-th element of matrix $\bTheta_{\varepsilon,j}$. Figure \ref{timeseries} visualizes dynamics of the precision matrix.	
	\vspace {20pt}
	\begin{figure}[!htb]
		\begin{center}
			\begin{tikzpicture}[snake=zigzag, line before snake = 5mm, line after snake = 5mm, box/.style = {inner xsep=0pt, outer sep=0pt,text width=0.15\linewidth,	align=left, font=\scriptsize}]
				\draw[line join=bevel] (0,0) -- (3,0);
				\draw[line join=bevel] (3,0) -- (4,0);
				\draw[snake] (4,0) -- (7,0);
				\draw[line join=bevel] (7,0) -- (8,0);
				
				\foreach \x in {0,3,8}
				\draw (\x cm,3pt) -- (\x cm,-3pt);
				
				\draw (0,0) node[below=3pt] {\scriptsize {$ t_1 $}} node[above=10pt, box] {\scriptsize $
					\text{tr}\Big(\widehat{\bSigma}_{\varepsilon,1}\bTheta_{\varepsilon,1}\Big)\newline-\log\det\bTheta_{\varepsilon,1} \newline+\alpha\norm{\bTheta_{\varepsilon,1}}_{\text{od},1}$};
				\draw (1,0) node[below=10pt] {\hspace{30pt}\scriptsize$ \beta\psi(\bTheta_{\varepsilon,2}-\bTheta_{\varepsilon,1}) $} node[above=3pt] {$  $};
				\draw (2,0) node[below=3pt] {$  $} node[above=3pt] {$  $};
				\draw (3,0) node[below=3pt] {\scriptsize {$ t_2 $}} node[above=10pt, box] {\scriptsize $
					\text{tr}\Big(\widehat{\bSigma}_{\varepsilon,2}\bTheta_{\varepsilon,2}\Big)\newline-\log\det\bTheta_{\varepsilon,2} \newline+\alpha\norm{\bTheta_{\varepsilon,2}}_{\text{od},1}$};
				\draw (4,0) node[below=10pt] {\hspace{10pt} \scriptsize$ \beta\psi(\bTheta_{\varepsilon,3}-\bTheta_{\varepsilon,2}) $ } node[above=3pt] {$  $};
				\draw (5,0) node[below=3pt] {$  $} node[above=3pt] {$  $};	
				\draw (6,0) node[below=3pt] {$ $} node[above=3pt] {$  $};
				\draw (7,0) node[below=10pt] {\hspace{-4pt} \scriptsize$ \beta\psi(\bTheta_{\varepsilon,N+1}-\bTheta_{\varepsilon,N}) $} node[above=3pt] {$  $};
				\draw (8,0) node[below=3pt] {\scriptsize {$ t_{N+1} $}} node[above=10pt, box] {\scriptsize $
					\text{tr}\Big(\widehat{\bSigma}_{\varepsilon,N+1}\bTheta_{\varepsilon,N+1}\Big)\newline-\log\det\bTheta_{\varepsilon,N+1} \newline+\alpha\norm{\bTheta_{\varepsilon,N+1}}_{\text{od},1}$};
			\end{tikzpicture}
			\caption{\textbf{Change of precision matrix over time:} $\beta$ is the penalty that enforces temporal consistency and $\psi$ is a convex penalty function.}
			\label{timeseries}
		\end{center}
	\end{figure}

	The optimization problem in \eqref{eq2} has two tuning parameters: $\alpha$, which determines the sparsity level of the network, and $\beta$, which controls the strength of resemblance between two neighboring precision estimators.\footnote{\cite{pesaran2006hierarchical} use Bayesian framework to handle parameter instability using hierarchical priors. As pointed out in their paper, ``intuition for the use of hierarchical priors comes from the shrinkage
	literature, since one can think of the parameters within the individual regimes as being shrunk
			towards a set of the so-called hyperparameters that characterize the ``top” layer of the hierarchy" (p. 1059).}
	In simulations and the empirical application we use the following procedure for tuning $\alpha$ and $\beta$: first, we set a grid of values $(\alpha,\beta) \in \{0,0.25,0.5,1,10,30\}$. Second, we use the first 2/3 of the training data to estimate forecast combination weights and jointly tune $\alpha$ and $\beta$ in the remaining 1/3 to yield the smallest value of the objective function, which is chosen to be either $\vertiii{\cdot}_2$-loss of precision matrix for simulations in Subsection 5.1, or MSFE for simulations in Subsection 5.2 and the empirical application. Note that when $\beta=0$, the optimization in \eqref{eq2} reduces to estimating $\bTheta_{\varepsilon,i}$ using Algorithm \ref{alg2} in each regime separately. Naturally, this incorporates the case when the structural break is strong and only the post-break data is used for producing forecast combination weights. When $\beta$ is large, there are weak structural breaks in $\bTheta_{\varepsilon,j}$, and  $\bTheta_{\varepsilon,j}$'s are estimated by using the data across different regimes. Section 6 provides more discussion on this in the context of our empirical application.
	
	The smoothing function $\psi(\cdot)$ in \eqref{eq2} can be LASSO ($\psi=\sum_{l,q}\abs{\cdot}$), Group LASSO ($\psi=\sum_{q}\norm{\cdot_q}_2$), or Ridge ($\psi=\sum_{l,q}(\cdot_{lq})^{2}$). LASSO penalty encourages small changes in the precision matrix over time: when the $lq$-th element changes at two consecutive times, the penalty forces the rest of the elements of the precision to remain the same. Group LASSO penalty allows the entire graph to restructure at some time points. This penalty is useful for anomaly detection, since it can identify structural changes in the network structure. Ridge penalty allows the network to change smoothly over time. This penalty is less strict than the LASSO penalty: instead of encouraging the graphs to be exactly the same, it allows smooth transitions. In our empirical application we use Ridge penalty to accommodate smooth transitions of precision over time.
	
	To estimate \eqref{eq2} we use the ADMM algorithm described in details in Supplementary Appendix \ref{appendix_admm}. Once $\bTheta_{\varepsilon,i}$ is estimated, we combine estimated factors, loadings and precision matrix of the idiosyncratic components using Sherman-Morrison-Woodbury formula to estimate the final precision matrix of forecast errors and use it to compute optimal forecast combination weights. We call the aforementioned procedure RD-FGL and summarize it in Algorithm \ref{alg3}.
\begin{algorithm}[H]
		\caption{RD-FGL}
		\label{alg3}
		\begin{algorithmic}[1]
			\STATE  Estimate $\{\bb_{i}\}_{i=1}^{p}$ and $\{\bf_t(\gamma_j)\}_{t=1}^{T}$ in \eqref{equ1_tv} using the weighted least squares problem in \eqref{wls}. Get $\widehat{\bSigma}_{f}$, $\widehat{\bTheta}_f$ and $\widehat{\bvarepsilon}_{t}(\gamma_j) = \be_{t} -  \widehat{\bB}(\gamma_j)\widehat{\bf}_{t}(\gamma_j)$.
			\STATE  Solve \eqref{eq2} using ADMM to get $\widehat{\bTheta}_{\varepsilon,j}$.
			\STATE  Use $\widehat{\bTheta}_{\varepsilon,j}$, $\widehat{\bTheta}_f$ and $\widehat{\bB}(\gamma_j)$ from Steps 1-2 to get $\widehat{\bTheta}_{j}(\gamma_j)=\widehat{\bTheta}_{\varepsilon,j}-\widehat{\bTheta}_{\varepsilon,j}\widehat{\bB}(\gamma_j)\lbrack\widehat{\bTheta}_f+\widehat{\bB}'(\gamma_j)\widehat{\bTheta}_{\varepsilon,j}\widehat{\bB}(\gamma_j)\rbrack^{-1}\widehat{\bB}'(\gamma_j)\widehat{\bTheta}_{\varepsilon,j}.$
			\STATE Use $\widehat{\bTheta}_{j}(\gamma_j)$ to get forecast combination weights $\widehat{\bw}_j(\gamma_j)=\frac{\widehat{\bTheta}_{j}(\gamma_j)\biota_p}{\biota_p'\widehat{\bTheta}_{j}(\gamma_j)\biota_p}$.
		\end{algorithmic} 
	\end{algorithm} 
\noindent We develop a scalable implementation of \eqref{eq2} for the RD-FGL in Algorithm \ref{alg3} through ADMM, which is extensively discussed in Supplementary Appendix \ref{appendix_admm}. ADMM is a  distributed convex optimization approach (\cite{BoydProximalAlgorithms}) that allows us to split the optimization problem in \eqref{eq2} into a series of subproblems. As pointed out in \cite{tvgl}, the scalability of ADMM comes from the improved runtime: to estimate a $p \times p$ matrix, the cost per iteration of ADMM is $\mathcal{O}(p^3)$ (which is the cost of an eigendecomposition of the $\bTheta$ step in Supplemental Appendix \ref{appendix_admm}). In contrast, the runtime of general interior-point methods is $\mathcal{O}(p^6)$ (\cite{mohan2014node}). 

\begin{remark}
Let us comment on the theoretical properties of RD-FGL. First, as shown in \cite{su2017time}, introducing time-varying factors does not change the main assumptions \ref{A1}-\ref{A3} on the errors, factors, factor loadings, and their interactions. This is because, as shown in \eqref{wls-modified}, the formulation with time-varying loadings can be reduced to the conventional PCA problem. Additional assumption that we need to impose is that $\E{\bf_t\bf'_t}$ is homogeneous over $t$. As discussed in Subsection 4.1,  this assumption is not restrictive. Second, we assume that the number of factors, $q$, and the number of forecasts, $p$, are not affected by the structural changes in loadings or idiosyncratic precision matrix. Allowing $p$ and $q$ to change is a straightforward extension and is left for future research. Third, assumptions \ref{B1}-\ref{B3} and assumption \ref{C1} are required to hold for each regime $j=1,\ldots,N+1$. Finally, we allow $s(\bTheta_{\varepsilon, j})=\mathcal{O}_P(s_{n_j})$ and $d(\bTheta_{\varepsilon, j})=\mathcal{O}_P(d_{n_j})$ to change for $j=1,\ldots,N+1$.  Let  $\omega_{n_j}\defeq \sqrt{\log p/n_j} +1/\sqrt{p}$. As long as $\varrho_{n_j}^{-1}\omega_{n_j}\xrightarrow{P}0$ and $\varrho_{n_j}d_{n_j}s_{n_j}\xrightarrow{P}0$ for each $j$, RD-FGL achieves the same rate as FGL in each regime:
		\begin{enumerate}
		\item[(i)]If $\varrho_{n_j}d_{n_j}^{2}s_{n_j}\xrightarrow{\text{P}}0$, RD-FGL consistently estimates forecast combination weights $\widehat{\bw}_j(\gamma_j)$ in Algorithm \ref{alg3}: $\norm{\widehat{\bw}_j(\gamma_j)-\bw_j}_1=\mathcal{O}_P\Big(\varrho_{n_j}d_{n_j}^2s_{n_j}\Big)=o_P(1)$.
		\item[(ii)] If $\varrho_{n_j}d_{n_j}s_{n_j}\xrightarrow{\text{P}}0$, FGL consistently estimates $\text{MSFE}(\bw_j,\bSigma_j)$: $	\abs{\frac{\text{MSFE}(\widehat{\bw}_j(\gamma),\widehat{\bSigma}_j)}{\text{MSFE}(\bw_j,\bSigma_j)} -1}=$\\$\mathcal{O}_P(\varrho_{n_j}d_{n_j}s_{n_j} )=o_P(1)$.
	\end{enumerate}
\end{remark} 
\subsection{Unknown Break Time and Number of Breaks}
The previous two subsections assumed that the number and location of breaks are known. We now relax these assumptions. First, assume that the number of breaks in factor loadings, $N_B$, and the number of breaks in idiosyncratic precision, $N_{\bTheta}$, are known and $N_B=N_{\bTheta}=1$, but their locations are unknown and might differ from each other.

To estimate the location of the break in factor loadings, we adapt the procedure in \cite{bai2020breakinloadings}. For a given break point in loadings, $T_{1}$, define the sum of squared residuals (SSR) as in \eqref{wls}:
\begin{equation}\label{wls_unknown}
	\text{SSR}(T_1) = (pT)^{-1}\sum_{i=1}^{p}\sum_{t=1}^{T}\big[ e_{it} - \bb'_{it}\bf_t  \big]^2 K_{\gamma t},
\end{equation}
where $K_{\gamma t} = \gamma \1{t \leq T_{1}} + \1{t > T_{1}}$ is a discrete kernel. The estimated break date is given by $\widehat{T}_{1}=\argmin_{1\leq T_{1}\leq T-1} \text{SSR}(T_{1})$.

To estimate the location of the break in $\bTheta_{\varepsilon}$ we use the procedure similar to \cite{bai2010breakinprecision}. Define $t_{1}$ to be a break point in $\bTheta_{\varepsilon}$. Recall, $n_j = t_j - t_{j-1}$, where $j=1,2$. Note that the number of observations in each regime depends on $t_1$: $n_1 = t_1 - t_{0}$ and $n_2 = t_2 - t_{1}$. For a given break point in idiosyncratic precision $t_1$, define the following objective function as in \eqref{eq2}:
	\begin{align} \label{eq2_unknown}
L(t_1)=\sum_{j=1}^{2}n_j\Big[\text{tr}\Big(\widehat{\bSigma}_{\varepsilon,j}\bTheta_{\varepsilon,j}\Big)-\log\det\bTheta_{\varepsilon,j}		\Big] 
	+\alpha\norm{\bTheta_{\varepsilon,j}}_{\text{od},1}
	+\beta \psi(\bTheta_{\varepsilon,2}-\bTheta_{\varepsilon,1}). 
\end{align}
The estimated break date is given by $\widehat{t}_{1}=\argmin_{1\leq t_{1}\leq T-1} L(t_{1})$.\footnote{As noted by \cite{bai1998estimating}, it is difficult to detect the break near the end of the sample. It is reasonable to expect that if break magnitude is reduced, the procedures in \eqref{wls_unknown} and \eqref{eq2_unknown} might not detect the break or detect it with a delay. \cite{smith2021break} point out that in such situations only few observations from the current regime are available to estimate the model parameters, leading to volatile and inaccurate forecasts. To address slow detection of breaks, \cite{smith2021break} exploit information in the cross-section to detect breaks more rapidly in real time. To address small sample size problem, \cite{smith2021break} adopt a Bayesian approach that uses economically motivated priors to shrink the parameters towards sensible values that rule out economically implausible values. As we discussed in Footnote 4, the Bayesian framework is alternative to using kernel-weighted observations.}

\indent When the number of breaks in either loadings is known and greater than 1 ($N_B>1$) and/or $N_{\bTheta}>1$, we can use the one-at-a-time approach as in \cite{bai2010breakinprecision}: the objective functions are identical to \eqref{wls_unknown} and \eqref{eq2_unknown}. The breaks are estimated sequentially. Once the first break is obtained, we split the sample at the estimated break point, resulting in two subsamples. A single break point in each subsample is estimated, but only one that achieves the smallest objective function (\eqref{wls_unknown} or \eqref{eq2_unknown})  is retained. If the number of breaks is equal to two, the procedure is stopped. Otherwise, we continue splitting into subsamples until all breaks are estimated.

If the number of breaks is unknown, we proceed as suggested in \cite{bai2010breakinprecision}: in the aforementioned one-at-a-time approach apply the test for existence of break point (\cite{bai2003multiplebreak}) to each subsample before estimating a break point.
	\section{Monte Carlo}
	We divide the simulation results into two subsections. In the first subsection we study the consistency of the FGL and RD-FGL for estimating precision matrix and the combination weights. In the second subsection we evaluate the out-of-sample forecasting performance of combined forecasts in terms of MSFE. We compare the performance of forecast combinations based on the factor models in Algorithms \ref{alg3}, \ref{alg2} with equal-weighted (EW) forecast combination\footnote{As pointed out by the referee, EW arises when the forecast errors follow a factor structure (one factor, homogeneous idiosyncratic variance). It can be viewed as one of ``factor-based" methods.}, and combinations that use GL without factor structure (Algorithm \ref{alg1a}). We examine the performance of RD-FGL for different specifications of the smoothing function $\psi(\cdot)$ as described in Subsection 4.2. LASSO penalty is denoted as $\ell_1$, Group LASSO as $\ell_g$, and Ridge as $\ell_2$. Similarly to the literature on graphical models, all exercises use 100 Monte Carlo simulations. We present simulation results with a structural break in both $\bB$ and $\bTheta_{\varepsilon}$. The results without a break, with break only in $\bTheta_{\varepsilon}$, and with multiple breaks can be found in Supplemental Appendix \ref{additional_sims1}.
	\subsection{Consistent Estimation of Forecast Combination Weights}
	We consider sparse Gaussian graphical models which may be fully specified by a precision matrix $\bTheta_0$. Therefore, the random sample is distributed as $\be_t=(e_{1t},\ldots,e_{pt})' \sim \mathcal{N}(0,\bSigma_{0})$, where $\bTheta_0=(\bSigma_{0})^{-1}$  for $t=1,\ldots, T, \ i=1,\ldots, p$. Let $\widehat{\bTheta}$ be the precision matrix estimator. We show consistency of the FGL in (i) the operator norm, $\vertiii{\widehat{\bTheta}-\bTheta_{0}}_{2}$, and (ii) in $\ell_1$-vector norm for the combination weights, $\norm{\widehat{\bw}-\bw}_1$, where $\bw$ is given by \eqref{eq13}.
	
	The forecast errors are assumed to have the following structure:
	\begin{align} \label{e51}
	\underbrace{\be_t}_{p \times 1}=\bB \underbrace{\bf_t}_{q\times 1}+ \ \bvarepsilon_t, \ \bf_{t}=\phi_f\bf_{t-1}+\bzeta_t,\quad t=1,\ldots,T
	\end{align}
	where $\be_{t}$ is a $p \times 1$ vector of forecast errors following $\mathcal{N}(\bm{0},\bSigma)$, $\bf_{t}$ is a $q \times 1$ vector of factors, $\bB$ is a $p\times q$ matrix of factor loadings, $\phi_f$ is an autoregressive parameter in the factors which is a scalar for simplicity, $\bzeta_t$ is a $q \times 1$ random vector with each component independently following  $\mathcal{N}(0,\sigma^{2}_{\zeta})$, $\bvarepsilon_t$ is a $p \times 1$ random vector following $\mathcal{N}(0,\bSigma_{\varepsilon})$, with sparse $\bTheta_{\varepsilon}$ that has a random graph structure described below. To create $\bB$ in \eqref{e51} we take the first $q$ columns of an upper triangular matrix from a Cholesky decomposition of the $p \times p$ Toeplitz matrix parameterized by $\rho$: that is, $\bB = (b)_{lm}$, where $(b)_{lm}=\rho^{\abs{l-m}}$, $l,m\in \{1,\ldots,p\}$.
	 We set $\rho = 0.2$, $\phi_f = 0.2$ and $\sigma^{2}_{\zeta} = 1$. The specification in \eqref{e51} leads to the low-rank plus sparse decomposition of the covariance matrix $\E{\be_t\be'_t}=\bSigma=\bB\bSigma_{f}\bB'+ \bSigma_{\varepsilon}$. When $\bSigma_{\varepsilon}$ has a sparse inverse $\bTheta_{\varepsilon}$, it leads to the low-rank plus sparse decomposition of the precision matrix $\bTheta$, such that $\bTheta$ can be expressed as a function of the low-rank $\bTheta_{f}$ plus sparse $\bTheta_{\varepsilon}$.
	
	We consider the following setup: let $p = T^{\delta}$, $\delta = 0.85$, $q = 2(\log(T))^{0.5}$ and $T = \lbrack 2^{\kappa} \rbrack, \ \text{for} \ \kappa=7,7.5,8,\ldots,9.5$. Our setup allows the number of individual forecasts, $p$, and the number of common factors in the forecast errors, $q$, to increase with the sample size, $T$.
	
	 A sparse precision matrix of the idiosyncratic components $\bTheta_{\varepsilon}$ is constructed as follows: we first generate the adjacency matrix using a random graph structure. Define a $p \times p$ adjacency matrix $\bA_{\varepsilon}$ which represents the structure of the graph
	with $a_{\varepsilon,lm}$ being the $l,m$-th element of the adjacency matrix $\bA_{\varepsilon}$. We set $a_{\varepsilon,lm} = a_{\varepsilon,ml}=1, \ \text{for} \ l\neq m$ with probability $\pi$, and $0$ otherwise. Such structure results in $s_T = p(p-1)\pi/2$ edges in the graph. To control sparsity, we set $\pi = 500/(pT^{0.8})$, which makes $s_T = \mathcal{O}(T^{0.05})$. The adjacency matrix has all diagonal elements equal to zero. To generate a sparse symmetric positive-definite precision matrix we use Scikit-Learn datasets package in Python (\cite{scikit}). To control the magnitude of partial correlations, the value of the smallest coefficient is set to 0.1 and the value of the largest coefficient is set to 0.3.
	
To incorporate structural breaks in $\bTheta_{\varepsilon}$ and factor loadings $\bB$, we proceed as follows. We fix a single break point in the middle of the sample size, $T/2$: in the precision matrix of the idiosyncratic errors before the break, referred to as $\bTheta_{\varepsilon,1}$, the value of the largest coefficient is set to 0.4; whereas in the precision matrix of the idiosyncratic errors after the break, $\bTheta_{\varepsilon,2}$, the value of the largest coefficient is set to 0.6. As a consequence, even though both matrices are still sparse, $\bTheta_{\varepsilon,2}$ has larger partial correlations. We use $\bTheta_{\varepsilon,1}$ and $\bTheta_{\varepsilon,2}$ to generate $\bvarepsilon_t$ in \eqref{e51}. For the structural break in factor loadings (which is assumed to happen at the same time as the structural change in $\bTheta_{\varepsilon}$), before the break we set $\rho_1=0.2$ in the Toeplitz matrix used to generate $\bB$ (i.e., $\bB = (b)_{lm}$, where $(b)_{lm}=\rho^{\abs{l-m}}$), and after the break we set $\rho_2=0.6$.

Figure \ref{fig1_revision} shows the averaged (over Monte Carlo simulations) errors of the estimators of the precision matrix $\bTheta$ and the optimal combination weight versus the sample size $T$ in the logarithmic scale (base 2). The estimate of the precision matrix of the EW forecast combination is obtained using the fact that diagonal covariance and precision matrices imply equal weights. To determine the values of the diagonal elements we use the shrinkage intensity coefficient calculated as the average of the eigenvalues of the sample covariance matrix of the forecast errors (see \cite{Ledoit2004}).

Figure \ref{fig1_revision} examines the performance when there are breaks in both $\bTheta_{\varepsilon}$ and $\bB$: accounting for the break significantly reduces the estimation error of precision matrix and combination weights. We report the results for the case when $\gamma$ is estimated using cross-validation ($\gamma = \hat{\gamma}$) (as discussed in Section 4). Supplemental Appendix \ref{additional_sims1} presents the results for the case when the break is only in $\bTheta_{\varepsilon}$.

\begin{figure}[!h]
	\centering
	\includegraphics[width=\textwidth]{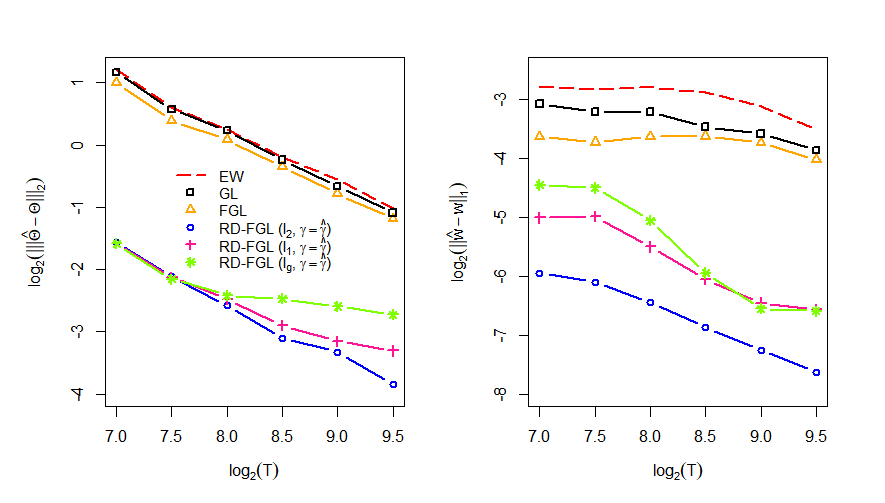} 
	\bigskip
	\caption{\textbf{Averaged errors of the estimators of $\bTheta$ (left) and $\bw$ on logarithmic scale (base 2): break in $\bTheta_{\varepsilon}$ and factor loadings $\bB$. $p = T^{0.85}$, $q = 2(\log(T))^{0.5}$, $s_T = \mathcal{O}(T^{0.05})$.}}
	\label{fig1_revision}
\end{figure}
	\subsection{Comparing Performance of Forecast Combinations}	
We consider the standard forecasting model in the literature (e.g., \cite{Stock2002}), which uses the factor structure of the high dimensional predictors.	
Suppose the data is generated from the following data generating process (DGP):
\begin{align}
&\bx_t=\bLambda\bg_t+\bv_t, \ \bg_{t}=\phi\bg_{t-1}+\bxi_t, \ y_{t+1}=\bg'_t\balpha+\sum_{s=1}^{\infty}\theta_{s}\epsilon_{t+1-s}+\epsilon_{t+1}, \label{e39}
\end{align}
where $y_{t+1}$ is a univariate series of our interest in forecasting, $\bx_t$ is an $M \times 1$ vector of regressors (predictors), $\balpha$ is an $M \times 1$ parameter vector, $\bg_{t}$ is an $r \times 1$ vector of factors, $\bLambda$ is an $M\times r$ matrix of factor loadings, $\bv_t$ is an $M \times 1$ random vector following $\mathcal{N}(0,\sigma^{2}_{v}\bI_M)$, $\phi$ is an autoregressive parameter in the factors which is a scalar for simplicity, $\bxi_t$ is an $M \times 1$ random vector with each component independently following $\mathcal{N}(0,\sigma^{2}_{\xi})$, $\epsilon_{t+1}$ is a random error following $\mathcal{N}(0,\sigma^{2}_{\epsilon})$, and $\balpha$ is an $r\times 1$ parameter vector which is drawn randomly from $\mathcal{N}(1,1)$. We set $\sigma_{\epsilon}=1$. The coefficients $\theta_{s}$ are set according to the rule $\theta_{s}=(1+s)^{c_1}c_{2}^{s}$
as in \cite{HANSEN2008}. We set $c_1=0.75$. We set $M=100$ and generate $r=5$ factors. To create $\bLambda$ in \eqref{e39} we take the first $r$ rows of an upper triangular matrix from a Cholesky decomposition of the $M \times M$ Toeplitz matrix parameterized by $\rho = 0.9$. The ranking of competing models was not very sensitive to varying values of $\phi$, $\rho$, $c_2$, and $r$.

One-step ahead forecasts are estimated from the factor-augmented autoregressive (FAR) models of orders $k,l$, denoted as  FAR($k,l$):
\begin{align}\label{e42}
&\hat{y}_{t+1}=\hat{\mu}
+\hat{\kappa}_1\hat{g}_{1,t}+\cdots+ \hat{\kappa}_k\hat{g}_{k,t}
+\hat{\psi}_1y_{t}+\cdots+\hat{\psi}_ly_{t+1-l},
\end{align}
where the factors $({\hat{g}_{1,t},\ldots,\hat{g}_{k,t}})$ are estimated from equation (\ref{e39}).
We consider the FAR models of various orders, with $k=1,\ldots,K$ and $l=1,\ldots,L$.
We also consider the models without any lagged $y$ or any factors. Therefore, the total number of forecasting models is $p \equiv (1+K)\times(1+L)$, which includes the forecasting models using naive average or no factors. We set $K=2$ and $L=7$.

The total number of observations is $m$. The period for training the models is set to be $m_1 = m/2$ -- this is used to train competing FAR models in \eqref{e42}. The remaining part of the sample, $m_2 = m-m_1$ is split as follows: the estimation window for training competing models (that is, EW, GL, FGL, and RD-FGL) is set to be of size $= m_2/2$. We roll the estimation window over the the test sample of the size $m_2/2$ to update all the estimates in each point of time. Recall that $q$ denotes the number of factors in the forecast errors as in equation (\ref{equ1}).

To incorporate structural break we proceed as follows. The period for training the models is set to be $m_1 = m/3$ -- this is used to train competing FAR models in \eqref{e42}. The remaining part of the sample, $m_2 = m-m_1$ is split as follows: the estimation window for training competing models is set to be of size  $= m_2/2$. We roll the estimation window over the test sample of the size $m_2/2$. The break point is fixed at 1/2 of the first estimation window. Before the break, when generating $\theta_{s}$ we set $c_2 = 0.3$, and after the break $c_2 = 0.9$. All other parameters stay unchanged. Notice that the break in $c_2$ can propagate into both a break in precision matrix and factor loadings.

Similarly to the previous subsection, we include different specifications of the smoothing function $\psi(\cdot)$. Figure \ref{fig6} shows the performance of all models including RD-FGL with $\gamma$ estimated using cross-validation: similarly to the conclusions in the previous subsection, accounting for the break significantly reduces MSFE of the combined forecast.
\begin{figure}[!htb]
	\centering
	\includegraphics[width=0.55\textwidth]{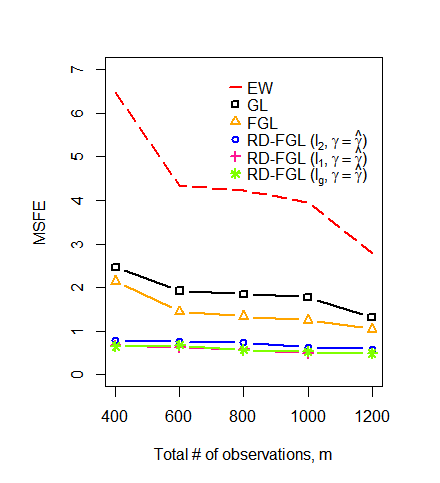} 
	\bigskip
	\caption{\textbf{Plots of the MSFE over the total number of observations \textit{m}}. $c_1=0.75$,  $c_2=0.3$ (before the break), $c_2=0.9$ (after the break), $ M=100, \ r=5, \sigma_\xi=1,\ L=7,\ K=2,\ p=24,\ q=3, \ \rho=0.9,\ \phi=0.8$.}
	\label{fig6}
\end{figure}
	\section{Application to Combining ECB SPF Forecasts}
	We use quarterly forecasts on the expected rates of inflation, real GDP growth and unemployment rate in the Euro area published by the \href{https://www.ecb.europa.eu/stats/ecb_surveys/survey_of_professional_forecasters/html/index.en.html}{ECB}. The raw data records 119 forecasters in total, but the panel is highly unbalanced with many missing values due to entry and exit in the long span. We follow \cite{Shi_ForecastCombinations} to obtain most qualified forecasters: first, we filter out irregular respondents if they missed more than 45\% of the observations; second, we use a random forest imputation algorithm (\cite{MissForest-R,MissForest-Paper}) to interpolate the remaining missing values. We consider the forecasts of three main economic indicators: (1) Real GDP growth defined as the year-on-year (YoY) percentage
	change of real GDP, based on standardized European System of National and Regional Accounts (ESA) definition. The time period under consideration is 1999:Q3-2023Q1 (which yields the total number of observations equal to 95), the final number of forecasters is $p=59$, and the prediction horizon is 2-quarters ahead. (2) Inflation which is defined as the YoY percentage change of the Harmonised Index of Consumer Prices (HICP) published by Eurostat. The time period under consideration is 1999:Q4-2023Q1 (which yields the total number of observations equal to 95),  the final number of forecasters is $p=56$, and the prediction horizon is 2-quarters ahead. (3) Unemployment rate which refers to Eurostat's definition and it is calculated as percentage of the labor force. The time period under consideration is 1999:Q3-2022Q4 (which yields the total number of observations equal to 94), the final number of forecasters is $p=45$, and the prediction horizon is 2-quarters ahead.
	
	We consider four choices of the training sample: $R\in\{20,30,40,50\}$, the estimation window is rolled over the test sample to update the estimates in each point of time. The optimal number of factors in the forecast errors (denoted as $q$ in equation \eqref{equ1}) is chosen using the standard data-driven method that uses the information criterion IC1 described in \cite{Bai2002}. In the majority of the cases the optimal number of factors was estimated to be equal to 1. To explore the benefits of using FGL and RD-FGL for forecast error quantification, we consider several alternative estimators of covariance/precision matrix of the idiosyncratic component in \eqref{3.11_main}: (1) linear shrinkage estimator of covariance developed by \cite{Ledoit2004} further referred to as Factor LW (FLW); (2) nonlinear shrinkage estimator of covariance by \cite{ledoit2017nonlinear} (Factor NLW or FNLW); (3) POET (\cite{fan2013POET}); (4) constrained $\ell_{1}$-minimization for inverse matrix estimator, CLIME (\cite{cai2011constrained}) (Factor CLIME or FCLIME); (5) nodewise regression developed by \cite{meinshausen2006} (Factor MB or FMB). To examine the benefits of imposing sparsity on $\bTheta_{\varepsilon}$ we also include the factor model without sparsity assumption on the idiosyncratic error precision matrix (referred to as Not Sparse) -- this corresponds to imposing $\tau=0$ in \eqref{e7.6}. To examine the benefits of using factor structure, we include several counterparts of the aforementioned models that directly estimate precision of the forecast errors without estimating factors and loadings: GL, LW, NLW, CLIME, and MB. For RD-FGL, similarly to the simulations, we include different specifications of the smoothing function $\psi(\cdot)$ and report the results for the case when the break parameter in factor loadings is estimated using cross-valiation ($\gamma=\hat{\gamma}$).

	 Our benchmark is the simple average with equal weights on all forecasters (referred to as EW). Going back to the discussion in Section 4 regarding setting $\beta=0$ in equation \eqref{eq2}: as we pointed out, this corresponds to using only post-break sample for estimation which is suboptimal since the value of  $\beta$ is already chosen optimally from the grid that includes $\beta = 0$ to minimize the MSFE. Hence, by construction, RD-FGL is superior to using only post-break data.
	
	For RD-FGL the number of breaks for loadings and precision is estimated using the test for existence of break point (\cite{bai2003multiplebreak}): using their sequential procedure we search for up to three breaks and
	set the trimming parameter to 10\% of the total number of observations,
	and the significance level at 5\%. The location of the break points for each series is estimated using the one-at-a-time approach described in Subsection 4.3.

	
	Table \ref{tab2} compares the performance of FGL and RD-FGL with the competitors for predicting three macroeconomic indicators for Euro-area using a combination of ECB SPF forecasts. It reports the ratios of MSFE of each method to the MSFE of the EW combined forecast. Using the Model Confidence Set (MCS) of \cite{hansenMCS}, we identify the set of superior models (SSM) for each series and horizon at 90\% confidence level. Once SSM is identified, we rank these models according to the relative sample loss of the $i$-th model relative to the average across models in SSM, and report the ranking in Table \ref{tab2}.

\begin{sidewaystable}
	\centering
	\resizebox{\linewidth}{!}{%
		\begin{tabular}{cccccccccccccccc} 
			\toprule
			& \textbf{GL} & \textbf{LW} & \textbf{NLW} & \textbf{CLIME} & \textbf{MB} & \textbf{POET} & \textbf{Not Sparse} & \textbf{FGL} & \textbf{FLW} & \textbf{FNLW} & \textbf{FCLIME} & \textbf{FMB} & \begin{tabular}[c]{@{}c@{}}\textbf{RD-FGL}\\~($\ell_2, \gamma=\hat{\gamma}$)\end{tabular} & \begin{tabular}[c]{@{}c@{}}\textbf{RD-FGL}\\~($\ell_1, \gamma=\hat{\gamma}$)\end{tabular} & \begin{tabular}[c]{@{}c@{}}\textbf{RD-FGL}\\~($\ell_g, \gamma=\hat{\gamma}$)\end{tabular} \\ 
			\hline
			\multicolumn{16}{c}{\textbf{Real GDP growth}} \\ 
			\hline
			$R$=20 & 3.3820 & 4.7845 & 1.5768 & 7.3391 & 1.1164 & 3.8275 & 7.0076 & 1.3178 & 0.9961 & 0.8719 & 3.4197 & 0.9962 & \textbf{0.4092} & 0.4142 & 0.4128 \\
			Ranking &  &  &  &  &  &  &  &  &  & 4 &  &  & 1 & 2 & 3 \\
			$R$=30 & 1.8984 & 1.2855 & 3.9910 & 5.9043 & 56.2033 & 0.9806 & 6.6966 & 0.9836 & 0.9372 & 3.5575 & 0.9865 & 0.8954 & \textbf{0.4163} & 0.4221 & 0.4204 \\
			Ranking &  &  &  &  &  &  &  &  &  &  &  &  & 1 & 2 & 3 \\
			$R$=40 & 2.1372 & 1.7982 & 4.9564 & 7.8310 & 0.9619 & 1.0327 & 18.1483 & 0.9353 & 0.9467 & 2.9197 & 0.9732 & 1.4725 & \textbf{0.4796} & 0.4868 & 0.4863 \\
			Ranking &  &  &  &  &  &  &  &  &  &  &  &  & 1 & 2 & 3 \\
			$R$=50 & 1.1706 & 1.4685 & 2.6014 & 3.9301 & 0.9594 & 1.0420 & 32.5874 & 0.9133 & 0.9399 & 2.5874 & 0.9860 & 0.9231 & \textbf{0.4296} & 0.4330 & 0.4329 \\
			Ranking &  &  &  &  &  &  &  &  &  &  &  &  & 1 & 2 & 3 \\ 
			\hline
			\multicolumn{16}{c}{\textbf{Inflation}} \\ 
			\hline
			$R$=20 & 0.9943 & 0.7184 & 0.5277 & 1.0063 & 2.9301 & 0.5431 & 1.7126 & 0.8970 & 0.6740 & 0.6617 & 0.5182 & 0.7917 & 0.3682 & 0.3711 & \textbf{0.3678} \\
			Ranking &  &  &  &  &  &  &  &  &  &  &  &  & 1 & 3 & 2 \\
			$R$=30 & 0.9659 & 0.7403 & 0.7628 & 1.0158 & 0.9019 & 0.5582 & 1.0544 & 0.9493 & 0.7301 & 0.5959 & 0.6970 & 0.7995 & \textbf{0.4326} & 0.4353 & 0.4832 \\
			Ranking &  &  &  &  &  & 4 &  &  &  &  &  &  & 1 & 3 & 2 \\
			$R$=40 & 0.8584 & 0.7780 & 0.5999 & 1.0122 & 0.9093 & 0.9463 & 1.0360 & 0.6641 & 0.5910 & 0.5692 & 0.6831 & 0.9416 & 0.3145 & 0.3355 & \textbf{0.3088} \\
			Ranking &  &  &  &  &  &  &  &  &  &  &  &  & 1 & 2 & 3 \\
			$R$=50 & 0.9014 & 0.7694 & 0.5576 & 0.6314 & 0.9375 & 1.0170 & 3.9136 & 0.5828 & 0.5903 & 0.5336 & 0.9573 & 0.7392 & \textbf{0.4184} & 0.4180 & 0.4314 \\
			Ranking &  &  &  &  &  &  &  &  &  &  &  &  & 1 & 2 & 3 \\ 
			\hline
			\multicolumn{16}{c}{\textbf{Unemployment rate}} \\ 
			\hline
			$R$=20 & 0.8731 & 0.9956 & 0.8496 & 0.9884 & 19.8034 & 0.9823 & 1.3067 & 0.9178 & 0.8732 & \textbf{0.7557} & 0.9247 & 1.6464 & 0.8951 & 0.9042 & 0.9054 \\
			Ranking & 4 &  &  &  &  & 6 &  & 5 & 3 & 2 & 1 &  &  &  &  \\
			$R$=30 & 0.8185 & 0.8763 & 0.8358 & 0.9713 & 42.1746 & 0.9815 & 1.7945 & \textbf{0.8031} & 0.8599 & 0.8320 & 0.9647 & 3.8414 & 1.0769 & 1.1357 & 1.1237 \\
			Ranking & 3 & 2 & 5 & 1 &  &  &  &  &  & 4 &  &  &  &  &  \\
			$R$=40 & 1.3733 & 1.3748 & 1.3526 & 1.4163 & 1.7992 & 0.9793 & 11.8839 & \textbf{0.9333} & 1.2844 & 1.2652 & 0.9941 & 1.1644 & 1.3359 & 1.3430 & 1.3383 \\
			Ranking* &  &  &  &  &  & 3 &  & 1 &  &  & 2 &  &  &  &  \\
			$R$=50 & 1.0145 & 1.4473 & 1.7531 & 1.2623 & 59.6932 & 1.4259 & 7.0196 & \textbf{0.8419} & 1.1558 & 1.1355 & 0.9896 & 2.7930 & 1.3761 & 1.4660 & 1.4607 \\
			Ranking & 3 &  &  &  &  &  &  & 1 &  &  & 2 &  &  &  &  \\
			\bottomrule
		\end{tabular}
	 }
\caption{\textbf{Prediction of Quarterly Macroeconomic Variables for Euro-area Using ECB SPF Forecasts.} MSFEs of competing methods are reported for each value of $R$, where $R$ indicates the length of the training window. Ratio indicates the ratio to MSFE of the Equal-Weighted combined forecast. Models with the lowest ratio are in bold. Models that belong to the SSM according to MCS test are ranked according to the relative sample loss of the $i$-th model relative to the average across models in SSM (at 90\% confidence level). For unemployment rate ($R$=40), the ranking marked with a star means that EW is included in the SSM (EW is ranked the 4th).} 
\label{tab2}
\end{sidewaystable}

	  There are three main findings that we learn from analyzing Table \ref{tab2}: \textbf{(1)} for all series factor-based models outperform non-factor ones. This means that incorporating the factor structure in the forecast errors improves forecasting performance. \textbf{(2)} for all series the Not Sparse model provides one of the worst performances. This means that the factor structure per se is not sufficient to achieve performance gains over EW, hence, it is necessary to impose sparsity on the precision matrix of the idiosyncratic components. \textbf{(3)} For real GDP growth and inflation series RD-FGL is always included in the SSM. For the unemployment rate, FGL outperforms RD-FGL. This result is supported by the behavior observed in the actual series: real GDP growth and inflation exhibit strong breaks following the global financial crisis and Covid pandemic, however this is not the case for the unemployment rate series that did not have strong breaks throughout the whole sample period.
	\section{Conclusions}
	In this paper we develop a unified framework to generalize network inference under a factor structure in the presence of structural breaks. We overcome the challenge of using graphical models under the factor structure and provide a simple approach that allows practitioners to combine a large number of forecasts when experts tend to make common mistakes. Using pre- and post-break data, our new approach to forecast combinations breaks down forecast errors into common and unique parts which improves the accuracy of the combined forecast. We allow the structural breaks to affect factor loadings and idiosyncratic precision matrix. For the ease of practical use we develop a scalable optimization procedure for RD-FGL, based on the ADMM. The empirical application to forecasting macroeconomic series using the data of the ECB Survey of Professional Forecasters shows that incorporating (i) factor structure in the forecast errors together with (ii) sparsity in the precision matrix of the idiosyncratic components and (iii) regime-dependent combination weights improves the performance of a combined forecast.

	\cleardoublepage
	
	\phantomsection
	
	\addcontentsline{toc}{section}{References}
	
	\setlength{\baselineskip}{14pt}
	\bibliographystyle{apalike}
	\bibliography{LeeSereginaForecasting}

\cleardoublepage
\renewcommand{\appendixpagename}{}
\renewcommand{\thesubsection}{A.\arabic{subsection}}
\vspace{-35pt}
\begin{appendices} 
	\renewcommand{\thesection}{\Alph{section}}
	\renewcommand{\thesubsection}{\Alph{section}.\arabic{subsection}}
	\renewcommand{\theequation}{\Alph{section}.\arabic{equation}}
	\captionsetup{%
		figurewithin=section,
		tablewithin=section
	}

\setcounter{algorithm}{0}
\renewcommand{\thealgorithm}{\Alph{section}.\arabic{algorithm}}
	\begin{spacing}{1}
\begin{center}
\Large{\textsc{Supplemental Appendix to \\ 
		``Combining Forecasts under Structural Breaks Using Graphical LASSO"}}	
\end{center}
\end{spacing}
\section{Graphical Lasso Algorithm} \label{appendixA0theor}
	\begin{spacing}{1.25}
Recall that we have $p$ competing forecasts of the univariate series $y_t$,  $t=1,\ldots,T$. Let $\be_t=(e_{1t},\ldots,e_{pt})' \sim \mathcal{N} (\mathbf{0}, \bSigma)$ be a $p \times 1$ vector of forecast errors. Assume they follow a Gaussian distribution.
The precision matrix $\bSigma^{-1}\defeq \bTheta$ contains information about partial covariances between the variables. For instance, if $\theta_{ij}$, which is the $ij$-th element of the precision matrix, is zero, then the variables $i$ and $j$ are conditionally independent, given the other variables.

Let $\bW$ be the estimate of $\bSigma$. Given a sample $\{\be_t\}_{t=1}^{T}$, let $\bS = (1/T)\sum_{t=1}^{T}(\be_t)(\be_t)'$ denote the sample covariance matrix, which can be used as a choice for $\bW$. Also, let $\widehat{\bGamma}^2\defeq \textup{diag}(\bW)$ and its $(i,j)$-th element is denoted as $\widehat{\gamma}_{ij}$. We can write down truncated Gaussian log-likelihood (up to constants) $l(\bTheta)=\log\det(\bTheta)-\text{tr}(\bW\bTheta)$. When $\bW=\bS$, the maximum likelihood estimator of $\bTheta$ is $\widehat{\bTheta}=\bS^{-1}$. The objective function associated with truncated Gaussian log-likelihood is also known as Bregman divergence and was shown to be applicable for non-Gaussian distributions (\cite{ravikumar2011}).

In the high-dimensional settings it is necessary to regularize the precision matrix, which means that some edges will be zero. A natural way to induce sparsity in the estimation of precision matrix is to add penalty to the maximum likelihood and use the connection between the precision matrix and regression coefficients to maximize the following penalized log-likelihood that weighs the variables by their scale:
\begin{align} \label{eq62}
	&\widehat{\bTheta}_{\tau}=\argmin_{\bTheta=\bTheta'}\textup{tr}(\bW\bTheta)-\log\det(\bTheta)+\tau\sum_{i\neq j}\widehat{\gamma}_{ii}\widehat{\gamma}_{jj}\abs{\theta_{ij}},
\end{align}
over positive definite symmetric matrices, where $\tau\geq0$ is a penalty parameter for the off-diagonal elements. We refer to the objective function in \eqref{eq62} as a ``weighted penalized log-likelihood". The subscript $\tau$ in $\widehat{\bTheta}_{\tau}$ means that the solution of the optimization problem in \eqref{eq62} will depend upon the choice of the tuning parameter. In order to simplify notation, we will omit the subscript.

Define the following partitions of $\bW$, $\bS$ and $\bTheta$:
\begin{equation} \label{eq43}
	\bW=\begin{pmatrix}
		\underbrace{\bW_{11}}_{(p-1)\times(p-1)}&\underbrace{\bw_{12}}_{(p-1)\times 1}\\\bw_{12}'&w_{22}
	\end{pmatrix}, \bS=\begin{pmatrix}
		\underbrace{\bS_{11}}_{(p-1)\times(p-1)}&\underbrace{\bs_{12}}_{(p-1)\times 1}\\\bs_{12}'&s_{22}
	\end{pmatrix}, \bTheta=\begin{pmatrix}
		\underbrace{\bTheta_{11}}_{(p-1)\times(p-1)}&\underbrace{\btheta_{12}}_{(p-1)\times 1}\\\btheta_{12}'&\theta_{22}
	\end{pmatrix}.
\end{equation}
Let $\bbeta\defeq -\btheta_{12}/\theta_{22}$. The idea of GL is to set $\bW= \bS+\tau\bI$ in \eqref{eq62} and combine the gradient of \eqref{eq62} with the formula for partitioned inverses to obtain the following $\ell_1$-regularized quadratic program
\begin{equation}\label{eq50}
	\widehat{\bbeta}=\argmin_{\bbeta\in \mathbb{R}^{p-1}}\Bigl\{ \frac{1}{2}\bbeta'\bW_{11}\bbeta-\bbeta'\bs_{12}+\sum_{i\neq j}\tau\widehat{\gamma}_{ii}\widehat{\gamma}_{jj}\norm{\bbeta}_1\Bigr\}.
\end{equation} 
As shown by \cite{GLASSO}, \eqref{eq50} can be viewed as a LASSO regression, where the LASSO estimates are functions of the inner products of $\bW_{11}$ and $s_{12}$. Hence, \eqref{eq62} is equivalent to $p$ coupled LASSO problems. Once we obtain $\widehat{\bbeta}$, we can estimate the entries $\bTheta$ using the formula for partitioned inverses. The weighted GL procedure is summarized in Algorithm \ref{alg1a}. 
\end{spacing}
\begin{spacing}{1}
	\begin{algorithm}[H]
		\caption{Weighted Graphical LASSO}
		\label{alg1a}
		\begin{algorithmic}[1]
			\STATE 	Initialize $\bW= \bS+\tau\bI$, with $w_{ii}=s_{ii}$. The diagonal of $\bW$ remains the same in what follows.
			\STATE Estimate a sparse $\bTheta$ using the following weighted Graphical LASSO objective function:
			\begin{equation}\nonumber
				\widehat{\bTheta}_{\tau}=\argmin_{\bTheta=\bTheta'}\textup{tr}(\bW\bTheta)-\log\det(\bTheta)+\tau\sum_{i\neq j}\widehat{\gamma}_{ii}\widehat{\gamma}_{jj}\abs{\theta_{ij}},
			\end{equation}
			over positive definite symmetric matrices.
			\STATE Repeat for $j=1,\ldots,p,1,\ldots,p,\ldots$ until convergence:
			\begin{itemize}
				\item Partition $\bW$ into part 1: all but the $j$-th row and column, and part 2: the $j$-th row and column.
				\item  Solve the score equations using the cyclical coordinate descent:
				\begin{equation*}
					\bW_{11}\bbeta-\bs_{12}+\tau\widehat{\gamma}_{ii}\widehat{\gamma}_{jj}\cdot\text{Sign}(\bbeta)=\mathbf{0}.
				\end{equation*}
				This gives a $(p-1) \times 1$ vector solution $\widehat{\bbeta}.$
				\item Update $\widehat{\bw}_{12}=\bW_{11}\widehat{\bbeta}$.
			\end{itemize}
			\STATE In the final cycle (for $i=1,\ldots,p$) solve for 
			\begin{equation*}
				\frac{1}{\widehat{\theta}_{22}}=w_{22}-\widehat{\bbeta}'\widehat{\bw}_{12}, \quad 
				\widehat{\btheta}_{12}=-\widehat{\theta}_{22}\widehat{\bbeta}.
			\end{equation*}
		\end{algorithmic} 
	\end{algorithm}	
\end{spacing}
	\begin{spacing}{1.25}
As was shown in \cite{GLASSO}, the estimator produced by Algorithm \ref{alg1a} is guaranteed to be positive definite. Furthermore, \cite{Sara2018} showed that Algorithm \ref{alg1a} is guaranteed to converge and produces consistent estimator of precision matrix under certain sparsity conditions.
\end{spacing}
\section{Factor Graphical LASSO} \label{appendix_fgl}
	\begin{spacing}{1.25}
	\begin{algorithm}[H]
		\caption{Factor Graphical LASSO (FGL)}
		\label{alg2}
		\begin{algorithmic}[1]
			\STATE Estimate factors, $\widehat{\bf}_t$, and factor loadings, $\widehat{\bB}$, using PCA. Obtain $\widehat{\bSigma}_{f}=\frac{1}{T}\sum_{t=1}^{T}\widehat{\bf}_{t}\widehat{\bf}_{t}^{'}$, $\widehat{\bTheta}_{f}=\widehat{\bSigma}_{f}^{-1}$, $\widehat{\bvarepsilon}_t = \be_t-\widehat{\bB}\widehat{\bf_t}$, and $\widehat{\bSigma}_{\varepsilon}=\frac{1}{T}\sum_{t=1}^{T}\widehat{\bvarepsilon}_{t}\widehat{\bvarepsilon}_{t}^{'}$.			
			\STATE Estimate a sparse $\bTheta_{\varepsilon}$ using the weighted Graphical LASSO in initialized with $\bW_{\varepsilon}=\widehat{\bSigma}_{\varepsilon}+\tau\bI$:
			\begin{align} \label{e7.6}
				&\widehat{\bTheta}_{\varepsilon,\tau}=\argmin_{\bTheta_{\varepsilon}=\bTheta'_{\varepsilon}}\text{tr}(\bW_{\varepsilon}\bTheta_{\varepsilon})-\log\det(\bTheta_{\varepsilon})+\tau\sum_{i\neq j}\widehat{\gamma}_{\varepsilon,ii}\widehat{\gamma}_{\varepsilon,jj}\abs{\theta_{\varepsilon,ij}}.
			\end{align}
			where $\widehat{\gamma}_{\varepsilon,ii}$ is the $(i,i)$-th element of $\widehat{\bGamma}_{\varepsilon}^2\defeq \textup{diag}(\bW_{\varepsilon})$.			
			\STATE Use $\widehat{\bTheta}_f$ from Step 1 and $\widehat{\bTheta}_{\varepsilon}$ from Step 2 to estimate $\bTheta$ using the Sherman-Morrison-Woodbury formula:
			\begin{equation}\label{3.11}
				\widehat{\bTheta}=\widehat{\bTheta}_{\varepsilon}-\widehat{\bTheta}_{\varepsilon}\widehat{\bB}\lbrack\widehat{\bTheta}_f+\widehat{\bB}'\widehat{\bTheta}_{\varepsilon}\widehat{\bB}\rbrack^{-1}\widehat{\bB}'\widehat{\bTheta}_{\varepsilon}.
			\end{equation}
		\end{algorithmic} 
	\end{algorithm}	
\end{spacing}
	\begin{spacing}{1.25}
Let $\widehat{\bTheta}_{\varepsilon,\tau}$ be the solution to \eqref{e7.6} for a fixed $\tau$. To choose the optimal shrinkage intensity coefficient, we minimize the following Bayesian Information Criterion (BIC) using grid search:
\begin{equation} \label{eq5.1}
	\text{BIC}(\tau) \defeq T\Big[\text{tr}(\widehat{\bTheta}_{\varepsilon,\tau}\widehat{\bSigma}_{\varepsilon})-\log\text{det}(\widehat{\bTheta}_{\varepsilon, \tau}) \Big] + (\log T)\sum_{i\leq j}\1{\widehat{\theta}_{\varepsilon,\tau,ij}\neq 0}.
\end{equation}
The grid $\mathcal{G}\defeq \{\tau_1,\ldots,\tau_M\}$ is constructed as follows: the maximum value in the grid, $\tau_{M}$, is set to be the smallest value for which all the off-diagonal entries of $\widehat{\bTheta}_{\varepsilon,\tau_{M}}$ are zero, that is, the maximum modulus of the off-diagonal entries of $\widehat{\bSigma}_{\varepsilon}$. The smallest value of the grid, $\tau_{1}\in \mathcal{G}$, is determined as $\tau_{1}\defeq\vartheta\tau_{M}$ for a constant $0<\vartheta<1$. The remaining grid values $\tau_1,\ldots,\tau_M$ are constructed in the ascending order from $\tau_{1}$ to $\tau_{M}$ on the log scale:
\begin{equation*}
	\tau_i=\exp \Big(\log(\tau_{1})+\frac{i-1}{M-1}\log(\tau_{M}/\tau_{1})  \Big), \quad i=2,\ldots,M-1.
\end{equation*}
We use $\vartheta= \sqrt{\log p/T} +1/\sqrt{p}$ (motivated by the convergence rate from Theorem \ref{theor1}) and $M=10$ in the simulations and the empirical exercise. 
\end{spacing}
	\section{Proof of Theorem 1} \label{appendixAtheor}
	We first present a lemma which is used in the proof.
	\begin{lem}\label{lemma1} 
		~
		\begin{enumerate}[label=(\alph*)]
			\item $\vertiii{\bTheta}_1=\mathcal{O}(d_T)$.
			\item $a \geq c>0$, where $a$ was defined in Section 3 and $c$ was defined in Assumption \textbf{(A.1)} (ii).
			\item $\abs{\widehat{a}-a}=\mathcal{O}_P(\varrho_{T}d_Ts_T)$, where $\widehat{a}$ was defined in Section 3.	
		\end{enumerate}
	\end{lem}
	\begin{proof}
		~
		\begin{enumerate}[label=(\alph*)]
			\item To prove part (a) we use the following matrix inequality which holds for any $\bA \in \mathcal{S}_p$:
			\begin{equation}\label{a3}
				\vertiii{\bA}_1=\vertiii{\bA}_{\infty}\leq \sqrt{d(\bA)}\vertiii{\bA}_2,
			\end{equation}
			where $d(\bA)$ was defined in Section 4. The proof of \eqref{a3} is a straightforward consequence of the Schwarz inequality.
			
			Sherman-Morrison-Woodbury formula together with \eqref{a3} and Assumptions \textbf{(B.1)-(B.3)} yield:
			\begin{align}\label{a4}
				\vertiii{\bTheta}_1 &\leq \vertiii{\bTheta_{\varepsilon}}_1+\vertiii{\bTheta_{\varepsilon}\bB\lbrack\bTheta_{f}+\bB'\bTheta_{\varepsilon}\bB \rbrack^{-1}\bB'\bTheta_{\varepsilon}}_1 \nonumber\\
				&=\mathcal{O}(\sqrt{d_T})+\mathcal{O}\Big(\sqrt{d_T}\cdot p \cdot \frac{1}{p} \cdot \sqrt{d_T}\Big) = \mathcal{O}(d_T).
			\end{align}
			\item Under Assumption \textbf{(A.1)}:
			\begin{equation*}
				a=\biota'_p\bTheta\biota_p/p\geq  c>0.
			\end{equation*}
			\item Using the H\"{o}lders inequality, we have
			\begin{align*} 
				\abs{\widehat{a}-a}=\abs{\frac{\biota'_{p}(\widehat{\bTheta}-\bTheta)\biota_p}{p}}\leq \frac{\norm{(\widehat{\bTheta}-\bTheta)\biota_p }_1\norm{\biota_p}_{\infty}}{p} &\leq \vertiii{ \widehat{\bTheta}-\bTheta}_1 \\
				&=\mathcal{O}_P(\varrho_{T}d_Ts_T)=o_P(1),
			\end{align*}
			where the last rate is obtained using the assumptions of Theorem 1.
		\end{enumerate}		
	\end{proof}	
	\subsection{Proof of Theorem 1}
	First, note that the forecast combination weight can be written as
	\begin{align*}
		\widehat{\bw}-\bw &= \frac{\Big((a\widehat{\bTheta}\biota_p) - (\hat{a}\bTheta\biota_p)  \Big)/p}{\hat{a}a}\\
		&=\frac{ \Big((a\widehat{\bTheta}\biota_p) - (a\bTheta\biota_p) + (a\bTheta\biota_p) - (\hat{a}\bTheta\biota_p) \Big)/p  }{\hat{a}a}.
	\end{align*}
	As shown in Callot et al. (2019), the above can be rewritten as
	\begin{align}\label{a1}
		\norm{\widehat{\bw}-\bw}_1\leq \frac{a\frac{\norm{(\widehat{\bTheta}-\bTheta)\biota_p}_1 }{p}+\abs{a-\widehat{a} }\frac{\norm{\bTheta\biota_p}_1}{p}}{\abs{\widehat{a}}a}.
	\end{align}
	Prior to bounding the terms in \eqref{a1}, we first present an inequality which is used in the derivations. Let $\bA \in \mathbb{R}^{p\times p}$ and $\bv \in \mathbb{R}^{p\times1}$. Also, let $\bA_{j}$ and $\bA'_{j}$ be a $p \times 1$ and $1\times p$ row and column vectors in $\bA$, respectively.
	\begin{align}\label{a2}
		\norm{\bA\bv}_1 &= \abs{\bA'_{1}\bv}+\ldots+\abs{\bA'_{p}\bv} \leq \norm{\bA_1}_1\norm{\bv}_{\infty}+\ldots+\norm{\bA_p}_1\norm{\bv}_{\infty}\\
		&=\Bigg(\sum_{j=1}^{p} \norm{\bA_j}_1 \Bigg)\norm{\bv}_{\infty} \leq p \max_{j}\abs{\bA_{j}}_1\norm{\bv}_{\infty}.\nonumber
	\end{align}
	H\"{o}lders inequality was used to obtain each inequality in \eqref{a2}. If $\bA \in \mathcal{S}_{p}$, then the last expression can be further reduced to $p \vertiii{\bA}_1\norm{\bv}_{\infty}$.
	
	Let us now bound the right-hand side of \eqref{a1}. In the numerator we have:
	\begin{align}\label{a5}
		\frac{\norm{(\widehat{\bTheta}-\bTheta)\biota_p}_1 }{p} \leq \vertiii{\bTheta}_1=\mathcal{O}_P(\varrho_{T}d_Ts_T),
	\end{align}
	the rates was derived in Lee and Seregina (2020), and the inequality follows from \eqref{a2}.
	\begin{equation}\label{a6}
		\frac{\norm{\bTheta\biota_p}_1}{p}\leq \vertiii{\bTheta}_1=\mathcal{O}(d_T),
	\end{equation}
	where the rate follows from Lemma \ref{lemma1} (a) and the inequality is obtained from \eqref{a2}. Combining \eqref{a5}, \eqref{a6}, and Lemma \ref{lemma1} (c) we get:
	\begin{align}
		a\frac{\norm{(\widehat{\bTheta}-\bTheta)\biota_p}_1 }{p}+\abs{a-\widehat{a} }\frac{\norm{\bTheta\biota_p}_1}{p}&= \mathcal{O}(1)\cdot\mathcal{O}_P(\varrho_{T}d_Ts_T)+\mathcal{O}_P(\varrho_{T}d_Ts_T)\cdot \mathcal{O}(d_T) \nonumber
		\\&= \mathcal{O}_P(\varrho_{T}d_{T}^{2}s_T)=o_P(1),
	\end{align}
	where the last equality holds under the assumptions of Theorem 1.
	
	For the denominator of \eqref{a1} it easy to see that $\abs{\widehat{a}}a=\mathcal{O}_P(1)$ using the results of Lemma \ref{lemma1} (b).
	
	For the MSFE part of Theorem 1, using Lemma \ref{lemma1} (b)-(c), we get
	\begin{align*}
		\abs{\frac{\hat{a}^{-1}}{a^{-1}}-1}=\frac{\abs{a-\hat{a}}}{\abs{\hat{a}}} = \mathcal{O}_P(\varrho_{T}d_Ts_T)=o_P(1),
	\end{align*}
	where the last rate is obtained using the assumptions of Theorem 1.
	\cleardoublepage
	\section{Implementation via ADMM Algorithm}\label{appendix_admm}
	To enable practical implementation of the RD-FGL, we develop an optimization procedure using ADMM algorithm to solve the convex optimization problem in \eqref{eq2}.
	
	First, we need to reformulate the unconstrained problem in \eqref{eq2} as a constrained problem which can be solved using ADMM:
	\begin{align}
		\{\widehat{\bTheta}_{\varepsilon,j}\}_{j=1}^{N+1}=\argmin_{	\{\bTheta_{\varepsilon,j}\}_{j=1}^{N+1}}&\sum_{j=1}^{N+1}n_j\Big[\text{tr}\Big(\widehat{\bSigma}_{\varepsilon,j}\bTheta_{\varepsilon,j}\Big)-\log\det\bTheta_{\varepsilon,j}		\Big] 
		+\alpha\norm{\bTheta_{\varepsilon,j}}_{\text{od},1}\\
		&+\beta\sum_{j=2}^{N+1}\psi(\bTheta_{\varepsilon,j}-\bTheta_{\varepsilon,j-1}) \nonumber\\
		\text{s.t.} \ &\bZ_{j,0}=\bTheta_{\varepsilon,j}, \ \text{for}\ j=1,\ldots,N+1\\
		&\Big(\bZ_{j-1,1},\bZ_{j,2}\Big)=\Big(\bTheta_{\varepsilon,j-1}, \bTheta_{\varepsilon,j}\Big), \ \text{for}\ j=2,\ldots,N+1.
	\end{align}
	Let $\bZ=\Bigl\{\bZ_{0},\bZ_{1},\bZ_{2}\Bigr\}=\Bigl\{\Big(\bZ_{1,0},\ldots,\bZ_{N+1,0}\Big),\Big(\bZ_{1,1},\ldots,\bZ_{N,1}\Big),\Big(\bZ_{2,2},\ldots,\bZ_{N+1,2}\Big)\Bigr\}$.\\
	
	\noindent Let $\bU=\Bigl\{\bU_{0},\bU_{1},\bU_{2}\Bigr\}=\Bigl\{\Big(\bU_{1,0},\ldots,\bU_{N+1,0}\Big),\Big(\bU_{1,1},\ldots,\bU_{N,1}\Big),\Big(\bU_{2,2},\ldots,\bU_{N,2}\Big)\Bigr\}$ be the scaled dual variable and $\rho > 0$ is the ADMM penalty parameter. Now we can use scaled ADMM to write down the augmented Lagrangian:
	\begin{align}
		\mathcal{L}_{\rho}(\bTheta_{\varepsilon},\bZ,\bU)&=\sum_{j=1}^{N+1}n_j\Big[\text{tr}\Big(\widehat{\bSigma}_{\varepsilon,j}\bTheta_{\varepsilon,j}\Big)-\log\det\bTheta_{\varepsilon,j}		\Big] 
		+\alpha\norm{\bZ_{j,0}}_{od,1}\\
		&+\beta\sum_{j=2}^{N+1}\psi(\bZ_{j,2}-\bZ_{j-1,1}) \nonumber\\
		&+\Big(\frac{\rho}{2}\Big)\sum_{j=1}^{N+1}\Bigg(\norm{\bTheta_{\varepsilon,j}-\bZ_{j,0}+\bU_{j,0}}_{F}^2-\norm{\bU_{j,0}}_F^2 \Bigg)\nonumber\\
		&+\Big(\frac{\rho}{2}\Big)\sum_{j=2}^{N+1}\Bigg(\norm{\bTheta_{\varepsilon,j-1}-\bZ_{j-1,1}+\bU_{j-1,1}}_{F}^2-\frac{\rho}{2}\norm{\bU_{j-1,1}}_F^2\nonumber\\
		&+\norm{\bTheta_{\varepsilon,j}-\bZ_{j,2}+\bU_{j,2}}_{F}^2-\norm{\bU_{j,2}}_F^2\Bigg).\nonumber
	\end{align}
	Let $k$ denote the iteration number, then ADMM consists of the following iterative updates:
	\begin{align}
		&\bTheta^{k+1}_{\varepsilon,j}\defeq \argmin_{\bTheta_{\varepsilon}\succ 0}\mathcal{L}_{\rho}(\bTheta_{\varepsilon},\bZ^{k},\bU^{k}), \label{eq43a}\\
		&\bZ^{k+1}=\begin{bmatrix}
			\bZ_{0}^{k+1}\\\bZ_{1}^{k+1}\\\bZ_{2}^{k+1}
		\end{bmatrix}\defeq \argmin_{\bZ_{0},\bZ_{1},\bZ_{2}}\mathcal{L}_{\rho}(\bTheta^{k+1}_{\varepsilon},\bZ,\bU^{k}),\label{eq44}\\
		&\bU^{k+1}=\begin{bmatrix}
			\bU_{0}^{k+1}\\\bU_{1}^{k+1}\\\bU_{2}^{k+1}
		\end{bmatrix}\defeq \begin{bmatrix}
			\bU_{0}^{k}\\\bU_{1}^{k}\\\bU_{2}^{k}
		\end{bmatrix}+\begin{bmatrix}
			\bTheta^{k+1}_{\varepsilon}-\bZ_{0}^{k+1}\\(\bTheta_{\varepsilon,1}^{k+1},\ldots,\bTheta_{\varepsilon,N}^{k+1})-\bZ_{1}^{k+1}\\(\bTheta_{\varepsilon,2}^{k+1},\ldots,\bTheta_{\varepsilon,N+1}^{k+1})-\bZ_{2}^{k+1}
		\end{bmatrix}.\label{eq45}
	\end{align}
	\newline \textbf{The $\bZ$ step:}\\
	The updating rule in \eqref{eq44} is easily recognized to be the element-wise soft thresholding operator. However, we need to split it into two updates since $(\bZ_{1},\bZ_{2})$ have to be updated jointly. Therefore, the update for $\bZ_{j,0}^{k+1}$ will be:
	\begin{align}
		&\bZ_{j,0}^{k+1}\defeq S_{\alpha/\rho}(\bTheta_{\varepsilon,j}^{k+1}+\bU_{j,0}^{k}),
	\end{align}
	where $S_{\alpha/\rho}(\cdot)$ is the element-wise soft-thresholding operator.\\
	
	\noindent We will solve a separate update for each $(\bZ_{j,2},\bZ_{j-1,1})$ pair for $j=2,\ldots,N+1$:
	\begin{align}\label{eq47}
		(\bZ_{j,2}^{k+1},\bZ_{j-1,1}^{k+1})&=\argmin_{\bZ_{j,2},\bZ_{j-1,1}}\Big(\frac{\rho}{2}\Big)\Bigg(\norm{\bTheta_{\varepsilon,j}-\bZ_{j,2}+\bU_{j,2}}_{F}^2 \\
		&+\norm{\bTheta_{\varepsilon,j-1}-\bZ_{j-1,1}+\bU_{j-1,1}}_{F}^2+\beta\psi(\bZ_{j,2}-\bZ_{j-1,1})\Bigg). \nonumber
	\end{align}
	Note that \eqref{eq47} is guaranteed to converge to a fixed point since it can be written as a proximal operator:
	\begin{equation} \label{eq48}
		(\bZ_{j,2}^{k+1},\bZ_{j-1,1}^{k+1})=\text{prox}_{\frac{\beta}{\rho}\psi(\cdot)}\Big(\bTheta_{\varepsilon,j}+\bU_{j,2},\bTheta_{\varepsilon,j-1}+\bU_{j-1,1} \Big)
	\end{equation}
	\begin{remark}
		A proximal operator of the scaled function $\nu f$, where $\nu>0$ can be expressed as:
		\begin{equation}\nonumber
			\text{prox}_{\nu f}(v)=\argmin_{x}\Big(	f(x)+\frac{1}{2\nu}\norm{x-v}^2_2\Big),
		\end{equation}
		where $f$ is a closed proper convex function. Note:
		\begin{equation}\nonumber
			\text{prox}_{\nu f}(v)\approx v-\tau \nabla f(v).
		\end{equation}
		\cite{BoydProximalAlgorithms} show that the fixed points of the proximal operator of $f$ are precisely the minimizers of $f$, i.e., $\text{
			prox}_{\nu f}(x^{\star})=x^{\star}$ if and only if $x^{\star}$ minimizes $f$.
	\end{remark}
	\noindent \textbf{The $\bTheta$ step:}\\
	The updating rule in \eqref{eq43a} can be further simplified to obtain a closed-form solution. Rewrite \eqref{eq43a}:
	\begin{align} \label{eq51}
		&\bTheta^{k+1}_{\varepsilon,j}=\argmin_{\Theta_{\varepsilon} \succ 0}\text{tr}\Big(\widehat{\bSigma}_j\bTheta_{\varepsilon,j}\Big)-\log\det\bTheta_{\varepsilon,j}+\frac{1}{2\eta}\norm{\bTheta_{\varepsilon,j}-\bA^{k}}_{F}^{2},
	\end{align}
	where $\bA^{k}=\dfrac{\bZ_{i,0}^{k}+\bZ_{j-1,1}^{k}+\bZ_{j,2}^{k}-\bU_{j,0}^{k}-\bU_{j-1,1}^{k}-\bU_{j,2}^{k}}{3}$, and $\eta=\frac{n_j}{3\rho}$.\\
	
	\noindent Take the gradient of the updating rule in \eqref{eq51} in order to get an analytical solution:
	\begin{align} 
		&\widehat{\bSigma}_{\varepsilon,j}-\bTheta^{-1,(k+1)}_{\varepsilon,j}+\frac{1}{\eta}\Big(\bTheta^{k+1}_{\varepsilon,j}-\bA^k	\Big)=0,\label{eq52}\\
		&\frac{1}{\eta}\bTheta^{k+1}_{\varepsilon,j}-\bTheta^{-1,(k+1)}_{\varepsilon}=\frac{1}{\eta}\bA^k-\widehat{\bSigma}_{\varepsilon,j}. \label{eq53}
	\end{align}
	Equation \eqref{eq53} implies that $\bTheta^{k+1}_{\varepsilon,j}$ and $\frac{1}{\eta}\bA^k-\widehat{\bSigma}_{\varepsilon,j}$ share the same eigenvectors.\\
	
	\noindent  Let $\bQ_j\bLambda_j \bQ_j'$ be the eigendecomposition of  $\frac{1}{\eta}\bA^k-\widehat{\bSigma}_{\varepsilon,j}$, where $\bLambda_j=\text{diag}(\lambda_{1,j},\ldots,\lambda_{p,j})$, and $\bQ_j'\bQ_j=\bQ_j\bQ_j'=\bI$.\footnote{Note that in practice we need to check that $\bA^k$ is symmetric. If it is not, then we can define $\widetilde{\bA}\defeq \frac{\bA^{k}+(\bA^{k})'}{2}$ and use it in the described algorithm instead of $\bA^k$. Since $\bTheta_j$ is symmetric the results will not be affected.} Pre-multiply \eqref{eq53} by $\bQ_j'$ and post-multiply it by $\bQ_j$:
	\begin{equation}
		\frac{1}{\eta} \widetilde{\bTheta}^{k+1}_{\varepsilon,j}-\widetilde{\bTheta}_{\varepsilon,j}^{-1,(k+1)}=\bLambda_j. \label{eq54}
	\end{equation}
	Now construct a diagonal solution of \eqref{eq54}:
	\begin{align}
		&\frac{1}{\eta}\tilde{v}_{i,j}-\dfrac{1}{\tilde{v}_{i,j}}=\lambda_{i,j},
	\end{align}
	where $\tilde{v}_{i,j}$ denotes the $i$-th eigenvalue of $\widetilde{\bTheta}_{\varepsilon,j}$. Solving for $\tilde{v}_{i,j}$ we get:
	\begin{equation}
		\tilde{v}_{i,j}=\dfrac{\lambda_{i,j}+\sqrt{\lambda_{i,j}^2+\frac{4}{\eta}}}{2\eta^{-1}}.
	\end{equation}
	Now we can calculate $\bTheta^{k+1}_{\varepsilon,j}$ which satisfies the optimality condition in \eqref{eq54}:
	\begin{equation}
		\bTheta^{k+1}_{\varepsilon,j}=\frac{1}{2\eta^{-1}}\bQ_j\Big(\bLambda_j+\sqrt{\bLambda_j^2+4\eta^{-1}\bI}\Big)\bQ_j'.
	\end{equation}
	Use the definition of $\eta=\frac{n_j}{3\rho}$:
	\begin{equation}\label{e97}
		\bTheta^{k+1}_{\varepsilon,j}=\frac{n_j}{6\rho}\bQ_j\Big(\bLambda_j+\sqrt{\bLambda_j^2+\frac{12\rho}{n_j}\bI}\Big)\bQ_j'.
	\end{equation}
	Step \eqref{e97} is the most computationally intensive task in the algorithm since the runtime of decomposing a $p \times p$ matrix is $\mathcal{O}(p^3)$. Also, note that compared to standard ADMM without smoothing penalty $\beta$, \eqref{e97} enforces stronger shrinkage. This is consistent with our motivation for the additional constraint - to smooth the estimator of precision matrix.
\newpage 	
\section{Additional Simulations} \label{additional_sims1}
	\begin{spacing}{1.25}
\subsection{No Break}
In this section we present simulation results that augment the results in Section 5 by assuming there is no break in the DGP.

First, we explore behavior of precision matrix and weights estimates. The setup is the same as in Subsection 5.1, but there is no break in $\bTheta_{\varepsilon}$: the value of the smallest coefficient is set to 0.1 and the value of the largest coefficient is set to 0.3.

Figure \ref{f18} shows the averaged (over Monte Carlo simulations) errors of the estimators of the precision matrix $\bTheta$ and the optimal combination weight versus the sample size $T$ in the logarithmic scale (base 2). For comparison, we include RD-FGL in all simulations. Since there is no break in the DGP, the tuning parameter for the factor loadings $\gamma = 1$ and the value of $\beta$ is estimated to be zero. Henceforth, all specifications of the smoothing function $\psi(\cdot)$ yield similar results and we only include one of them RD-FGL ($\ell_2$). As evidenced by Figure \ref{f18}, FGL and RD-FGL demonstrate superior performance over EW and non-factor based model (GL). FGL and RD-FGL have comparable performance, but since there is no break in DGP, FGL is more efficient. Furthermore, FGL and RD-FGL achieve lower estimation error in the combination weights, which leads to lower risk of the combined forecast. Also, note that the precision matrix estimated using the EW method also shows good convergence properties.
\begin{figure}[!htbp]
	\centering
	\includegraphics[width=0.95\textwidth]{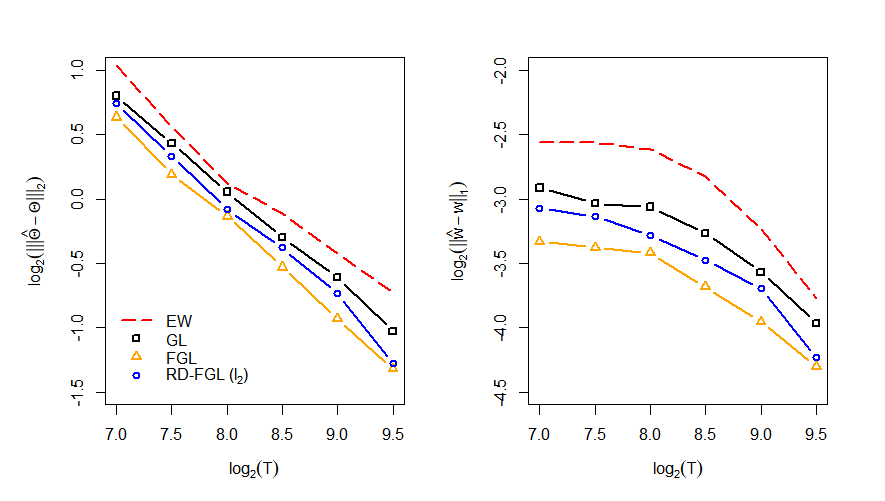} 
	\bigskip
	\caption{\textbf{Averaged errors of the estimators of $\bTheta$ (left) and $\bw$ on logarithmic scale (base 2). $p = T^{0.85}$, $q = 2(\log(T))^{0.5}$, $s_T = \mathcal{O}(T^{0.05})$.}}
	\label{f18}
\end{figure}

Second, we explore behavior of MSFE under no breaks. We set $c_1\in\{0,0.75\}$ and $c_2 = 0.9$. We set $M=100$ and generate $r=5$ factors. To create $\bLambda$ in \eqref{e39} we take the first $r$ rows of an upper triangular matrix from a Cholesky decomposition of the $M \times M$ Toeplitz matrix parameterized by $\rho = 0.9$. The ranking of competing models was not very sensitive to varying values of $\phi$, $\rho$, $c_2$, and $r$ -- the results examining sensitivity to a grid of 10 different AR(1) coefficients $\phi$ equidistant between $0$ and $0.9$, a grid of 10 different values of $\rho$ equidistant between $0$ and $0.9$,  $c_2 \in \{0.6, 0.7, 0.8, 0.9\}$, and $r \in  \{1,\cdots,7\}$ are available upon request.

One-step ahead forecasts are estimated from the factor-augmented autoregressive (FAR) models of orders $k,l$, denoted as  FAR($k,l$), defined in \eqref{e42}. We consider the FAR models of various orders, with $k=1,\ldots,K$ and $l=1,\ldots,L$.
We also consider the models without any lagged $y$ or any factors. Therefore, the total number of forecasting models is $p \equiv (1+K)\times(1+L)$, which includes the forecasting models using naive average or no factors. We set $K=2$ and $L=7$.

The total number of observations is $m$. The period for training the models is set to be $m_1 = m/2$ -- this is used to train competing FAR models in \eqref{e42}. The remaining part of the sample, $m_2 = m-m_1$ is split as follows: the estimation window for training competing models (that is, EW, GL, FGL, and RD-FGL) is set to be window $= m_2/2$. We roll the estimation window over the the test sample of the size $m_2/2$ to update all the estimates in each point of time. Recall that $q$ denotes the number of factors in the forecast errors as in equation (\ref{equ1}).

Similarly to the previous subsection, we include RD-FGL in all simulations. When there is no break in the DGP, the tuning parameter for the factor loadings, $\gamma$, is set to one, and the penalty that controls the change of idiosyncratic precision matrix over time, $\beta$, is zero. Figure \ref{fig2} shows the MSFE for different sample sizes and fixed parameters: we report the results for two values of $c_1\in\{0,0.75\}$. As evidenced from Figure \ref{fig2}, the models that use the factor structure outperform EW combination and non-factor based counterparts for both values of $c_1$.
\begin{figure}[!htb]
	\centering
	\includegraphics[width=\textwidth]{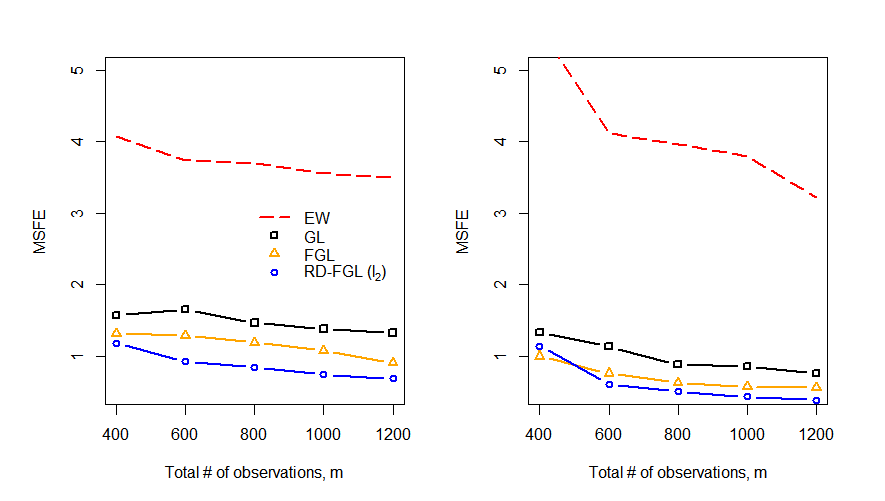} 
	\bigskip
	\caption{\textbf{Plots of the MSFE over the total number of observations \textit{m}}. $c_1=0$ (left), $c_1=0.75$ (right), $c_2=0.9, \ M=100, \ r=5, \sigma_\xi=1,\ L=7,\ K=2,\ p=24,\ q=3, \ \rho=0.9,\ \phi=0.8$.}
	\label{fig2}
\end{figure}
\subsection{Break Only in Idiosyncratic Precision Matrix}
This section presents the results for the case when there is a single break in $\bTheta_{\varepsilon}$. The DGP is the same as described in Subsection 5.1: the break point is fixed in the middle of the sample $T/2$. Before the break, the value of the largest coefficient in $\bTheta_{\varepsilon,1}$ is set to 0.4, after the break it changes to 0.6.

Figure \ref{fig1_rd} shows the averaged (over Monte Carlo simulations) errors of the estimators of the precision matrix $\bTheta$ and the optimal combination weight versus the sample size $T$ in the logarithmic scale (base 2). The estimate of the precision matrix of the EW forecast combination is obtained using the fact that diagonal covariance and precision matrices imply equal weights. To determine the values of the diagonal elements we use the shrinkage intensity coefficient calculated as the average of the eigenvalues of the sample covariance matrix of the forecast errors (see \cite{Ledoit2004}).

Figure \ref{fig1_rd} shows the performance of all models including RD-FGL when the break is only in $\bTheta_{\varepsilon}$ ($\gamma = 1$): accounting for the break significantly reduces the estimation error of precision matrix and combination weights.
\begin{figure}[!h]
	\centering
	\includegraphics[width=\textwidth]{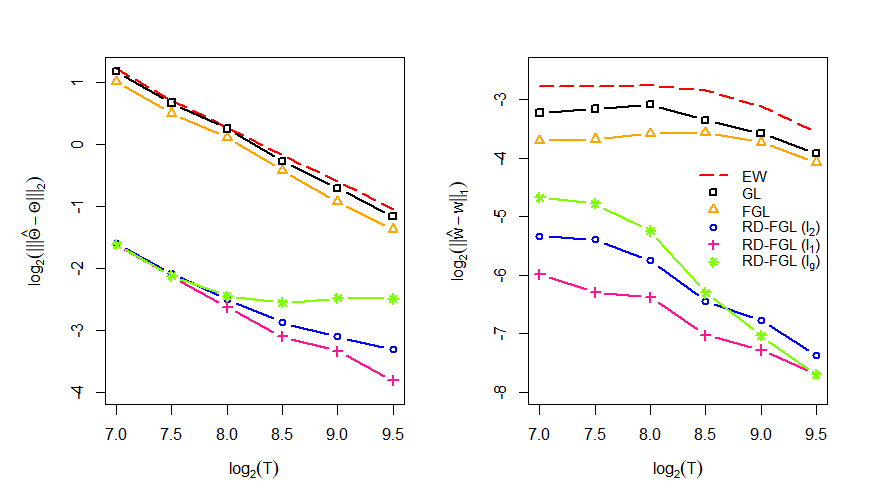}  
	\bigskip
	\caption{\textbf{Averaged errors of the estimators of $\bTheta$ (left) and $\bw$ on logarithmic scale (base 2): break in $\bTheta_{\varepsilon}$. $p = T^{0.85}$, $q = 2(\log(T))^{0.5}$, $s_T = \mathcal{O}(T^{0.05})$.}}
	\label{fig1_rd}
\end{figure}
\subsection{Multiple Breaks}
We examine the performance of RD-FGL and competing methods for the case of two known breaks.

First, we explore behavior of precision matrix and weights estimates. To incorporate two structural breaks in $\bTheta_{\varepsilon}$, we add the following modification to the DGP setup in Subsection 5.1. We fix two break points: one at $t_1=T/4$ and the other at $t_1=3T/4$. Define the following idiosyncratic precision matrices: $\bTheta_{\varepsilon,1}$ before $t_1$, $\bTheta_{\varepsilon,2}$ between $t_1$ and $t_2$, and $\bTheta_{\varepsilon,3}$ after $t_2$. The value of the largest coefficient in the three aforementioned matrices is set to 0.2, 0.4, and 0.6, accordingly.

\begin{figure}[!h]
	\centering
	\includegraphics[width=\textwidth]{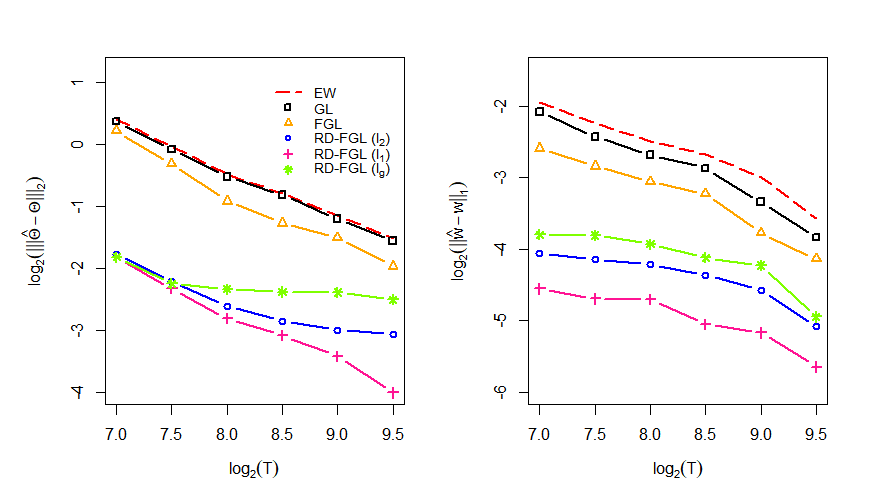} 
	\bigskip
	\caption{\textbf{Averaged errors of the estimators of $\bTheta$ (left) and $\bw$ on logarithmic scale (base 2): two breaks in $\bTheta_{\varepsilon}$. $p = T^{0.85}$, $q = 2(\log(T))^{0.5}$, $s_T = \mathcal{O}(T^{0.05})$.}}
	\label{fig1_app}
\end{figure}

As demonstrated in Figure \ref{fig1_app}, similarly to the findings in the main manuscript for the case with one break, accounting for the break significantly reduces the estimation error of precision matrix and combination weights.

Second, we explore behavior of MSFE under two breaks. To incorporate two structural breaks we add the following modification to the DGP in Subsection 5.2. The total number of observations is $m$. The period for training the models is set to be $m_1 = T/3$ -- this is used to train competing FAR models in \eqref{e42}. The remaining part of the sample, $m_2 = m-m_1$ is split similarly to Subsection 5.2: the estimation window for training competing models is set to be window $= m_2/2$. We roll the estimation window over the the test sample. The break points are fixed at 1/3 and 3/4 of the first estimation window, and will be referred to as $t_1$ and $t_2$.

When generating $\theta_{s}$ we set $c_2$ as follows: $c_2=0.3$ before $t_1$, $c_2=0.6$ between $t_1$ and $t_2$, $c_3=0.9$ after $t_2$.

Before the break, when generating $\theta_{s}$ we set $c_2 = 0.3$, and after the break $c_2 = 0.9$. All other parameters stay unchanged. Notice that the break in $c_2$ can propagate into both a break in precision matrix and factor loadings.

Similarly to the main manuscript, we include different specifications of the smoothing function $\psi(\cdot)$. Figure \ref{fig2_app} shows the performance of all models including RD-FGL with $\gamma$ estimated using cross-validation: similarly to the conclusions in Subsection 5.2, accounting for the break significantly reduces MSFE of the combined forecast.

\begin{figure}[!htb]
	\centering
	\includegraphics[width=0.55\textwidth]{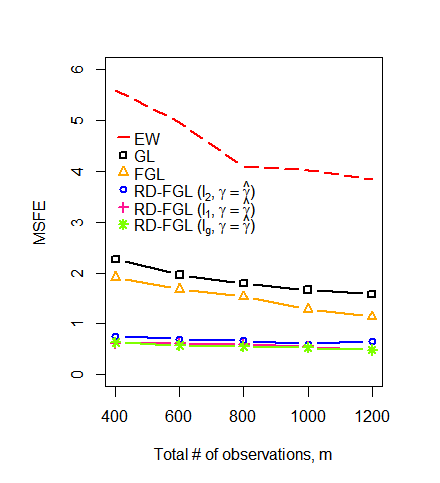} 
	\bigskip
	\caption{\textbf{Plots of the MSFE over the total number of observations \textit{m}}. $c_1=0.75$,  $c_2=0.3$ (before $t_1$), $c_2=0.6$ (between $t_1$ and $t_2$), $c_3=0.9$ (after $t_2$), $ M=100, \ r=5, \sigma_\xi=1,\ L=7,\ K=2,\ p=24,\ q=3, \ \rho=0.9,\ \phi=0.8$.}
	\label{fig2_app}
\end{figure}

\subsection{Varying Break Magnitude}
We examine the performance of RD-FGL and competing methods for the case of one known break of smaller magnitude.

First, we explore behavior of precision matrix and weights estimates. The setup is the same as in Subsection 5.1: we fix a single break point in the middle of the sample size, $T/2$: in the precision matrix of the idiosyncratic errors before the break, referred to as $\bTheta_{\varepsilon,1}$, the value of the largest coefficient is set to 0.4; whereas in the precision matrix of the idiosyncratic errors after the break, $\bTheta_{\varepsilon,2}$, the value of the largest coefficient is set to 0.45. We use $\bTheta_{\varepsilon,1}$ and $\bTheta_{\varepsilon,2}$ to generate $\bvarepsilon_t$ in \eqref{e51}.

\begin{figure}[!h]
	\centering
	\includegraphics[width=\textwidth]{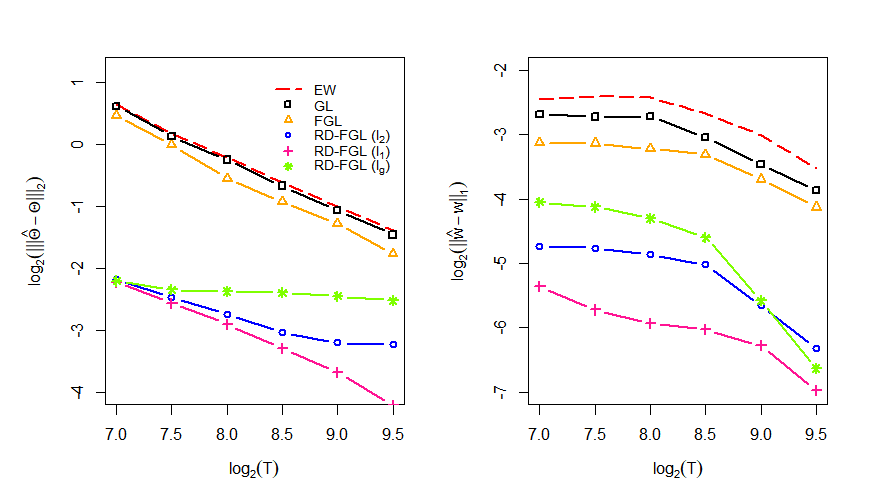} 
	\bigskip
	\caption{\textbf{Averaged errors of the estimators of $\bTheta$ (left) and $\bw$ on logarithmic scale (base 2): one break in $\bTheta_{\varepsilon}$. $p = T^{0.85}$, $q = 2(\log(T))^{0.5}$, $s_T = \mathcal{O}(T^{0.05})$.}}
	\label{fig3_app}
\end{figure}

As demonstrated in Figure \ref{fig3_app}, similarly to the findings in the main manuscript, accounting for the break significantly reduces the estimation error of precision matrix and combination weights even if the break magnitude is small.

Second, we explore behavior of MSFE for smaller break magnitude. The setup is the same as in Subsection 5.1:  the period for training the models is set to be $m_1 = m/3$ -- this is used to train competing FAR models in \eqref{e42}. The remaining part of the sample, $m_2 = m-m_1$ is split as follows: the estimation window for training competing models is set to be window $= m_2/2$. We roll the estimation window over the test sample of the size $m_2/2$. The break point is fixed at 1/2 of the first estimation window. Before the break, when generating $\theta_{s}$ we set $c_2 = 0.3$, and after the break $c_2 = 0.4$. All other parameters stay unchanged.

Similarly to the main manuscript, we include different specifications of the smoothing function $\psi(\cdot)$. Figure \ref{fig4_app} shows the performance of all models including RD-FGL with $\gamma$ estimated using cross-validation: similarly to the conclusions in Subsection 5.2, accounting for the break significantly reduces MSFE of the combined forecast even if the break magnitude is small.

\begin{figure}[!htb]
	\centering
	\includegraphics[width=0.55\textwidth]{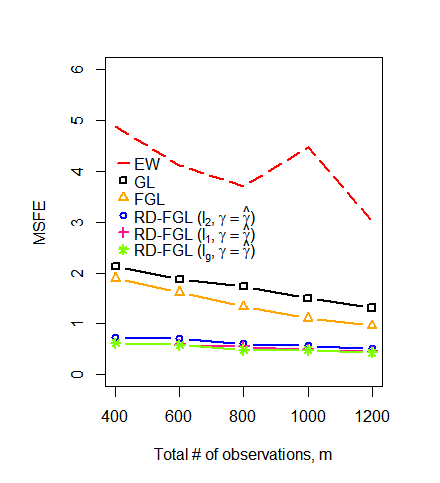} 
	\bigskip
	\caption{\textbf{Plots of the MSFE over the total number of observations \textit{m}}. $c_1=0.75$,  $c_2=0.3$ (before the break), $c_2=0.4$ (after the break), $ M=100, \ r=5, \sigma_\xi=1,\ L=7,\ K=2,\ p=24,\ q=3, \ \rho=0.9,\ \phi=0.8$.}
	\label{fig4_app}
\end{figure}
\end{spacing}
\end{appendices}
\end{document}